\documentclass[a4paper,UKenglish,cleveref, autoref]{lipics-v2019}

\title{Efficient parameterized algorithms for computing all-pairs shortest paths} 

\titlerunning{Efficient parameterized algorithms for computing all-pairs shortest paths}

\author{Stefan Kratsch}{Humboldt-Universit{\"a}t zu
Berlin, Germany}{kratsch@informatik.hu-berlin.de}{https://orcid.org/0000-0002-0193-7239}{}

\author{Florian Nelles}{Humboldt-Universit{\"a}t zu
Berlin, Germany}{nelles@informatik.hu-berlin.de}{}{}

\authorrunning{S. Kratsch and F. Nelles}

\Copyright{Stefan Kratsch and Florian Nelles}



\category{}

\relatedversion{}

\supplement{}

\nolinenumbers 

\hideLIPIcs  

\EventEditors{John Q. Open and Joan R. Access}
\EventNoEds{2}
\EventLongTitle{42nd Conference on Very Important Topics (CVIT 2016)}
\EventShortTitle{CVIT 2016}
\EventAcronym{CVIT}
\EventYear{2016}
\EventDate{December 24--27, 2016}
\EventLocation{Little Whinging, United Kingdom}
\EventLogo{}
\SeriesVolume{42}
\ArticleNo{23}

\usepackage{xspace}
\usepackage{tikz}

\newcommand{\prob}[1]{\textsc{\lowercase{#1}}\xspace}
\newcommand{\Oh}{\mathcal{O}} 
\newcommand{\NP}{\ensuremath{\textsf{NP}}\xspace}

\newcommand{\APSP}{\prob{All-Pairs Shortest Paths}}
\newcommand{\VWAPSP}{\prob{Vertex-Weighted All-Pairs Shortest Paths}}

\DeclareMathOperator{\cw}{\mathsf{cw}}
\DeclareMathOperator{\mw}{\mathsf{mw}}

\DeclareMathOperator*{\argmin}{arg\,min}

\begin{document}

\maketitle

\begin{abstract}
Computing all-pairs shortest paths is a fundamental and much-studied problem with many applications. Unfortunately, despite intense study, there are still no significantly faster algorithms for it than the $\Oh(n^3)$ time algorithm due to Floyd and Warshall (1962). Somewhat faster algorithms exist for the vertex-weighted version if fast matrix multiplication may be used. Yuster (SODA 2009) gave an algorithm running in time $\Oh(n^{2.842})$, but no combinatorial, truly subcubic algorithm is known.

Motivated by the recent framework of efficient parameterized algorithms (or ``FPT in P''), we investigate the influence of the graph parameters clique-width ($\cw$) and modular-width ($\mw$) on the running times of algorithms for solving \APSP. We obtain efficient (and combinatorial) parameterized algorithms on non-negative vertex-weighted graphs of times $\Oh(\cw^2n^2)$, resp. $\Oh(\mw^2n + n^2)$. If fast matrix multiplication is allowed then the latter can be improved to $\Oh(\mw^{1.842}n + n^2)$ using the algorithm of Yuster as a black box.
The algorithm relative to modular-width is adaptive, meaning that the running time matches the best unparameterized algorithm for parameter value $\mw$ equal to $n$, and they outperform them already for $\mw \in \Oh(n^{1 - \varepsilon})$ for any $\varepsilon > 0$.
\end{abstract}

\section{Introduction}

\APSP (APSP) is a fundamental and much-studied problem in the field of algorithmic graph theory. Next to the theoretical interest in the problem, \APSP is important for many practical applications, e.g., it is closely related to several vertex centrality measures in networks (for example, the betweenness centrality of a vertex $v$ is defined as the sum of the fraction of all-pairs shortest paths that pass through $v$). The \APSP problem is also considered as the core of many routing problems and has applications for example in areas such as routing protocols, driving direction on web mappings, transportation, and traffic assignment problems, and many more. See also the survey of Susmita~\cite{Susmita15} for more applications.

Despite the large interest in \APSP, there are only small improvements known since the well-known $\Oh(n^3)$-time algorithm by Floyd and Warshall~\cite{Floyd62, Warshall62} from 1962: Chan~\cite{Chan10} as well as Han and Takaoka~\cite{HanT16} gave an algorithm running in $\Oh(n^3 / \log^2 n)$ (omitting $poly(\log \log n)$ factors) and Williams~\cite{Williams18} gave an randomized algorithm running in time $\Oh(n^3 / 2^{\Omega(\log n)^{1/2}})$.
 While there are no unconditional lower bounds known, it has been conjectured that there is no truly subcubic algorithm for \APSP, i.e., that no algorithm achieves time $\Oh(n^{3-\varepsilon})$ for any $\varepsilon>0$. Using suitable subcubic reductions, this is tightly connected to the existence of subcubic algorithms for several network centrality measures, finding a directed triangle of negative total edge length, finding the second shortest simple path between two nodes in an edge-weighted graph, or checking if a given matrix defines a metric. This means that if one of those problems can be solved in truly subcubic time (i.e., can be solved in time $\Oh(n^{3-\varepsilon}) \cdot poly(\log M)$ for an $\varepsilon > 0$ and weights in $[-M, M]$ for weighted problems), then all of the problems admit algorithms with truly subcubic running time~\cite{WilliamsW18}.
The situation is different for \VWAPSP: While it is conjectured that there is no truly subcubic \emph{combinatorial} algorithm, faster algorithms are known if fast matrix multiplication may be used. The currently fastest algorithm is due to Yuster~\cite{Yuster09} and runs in time $\Oh(n^{2.842})$. For sparse graphs there is an algorithm running in time $\Oh(nm + n^2 \log\log n)$ for directed graphs~\cite{Pettie04} and an algorithm for undirected graphs~\cite{PettieR05} with a running time of  $\Oh(mn \log \alpha(m, n))$, where $\alpha$ is the inverse of the Ackermann’s function.

Independently of whether one believes in conditional lower bounds and hypotheses, the fact remains that we do not know any truly subcubic algorithms for \APSP nor truly subcubic, combinatorial algorithms for the vertex-weighted case. Besides heuristics or approximation algorithms, one possible solution for faster algorithms for at least \emph{some} input graphs is to exploit \emph{structure} in the input graph. In addition to measuring the complexity of a problem relative to the input size of a graph (number $n$ of vertices and number $m$ of edges), one may additionally consider some \emph{parameter}, say $k$, that quantifies structure that may be exploited by an algorithm; i.e., we may study the \emph{parameterized complexity} of the problem. This framework typically aims at \NP-hard problems and a key goal is to obtain \emph{fixed-parameter tractable} (FPT) algorithms that run in time $f(k)n^c$ for some constant $c$ and some (usually exponential) function $f(k)$ of the parameter. Initiated by the work of Giannopoulou et al.~\cite{GiannopoulouMN15}, also efficient parameterized algorithms for tractable problems are considered (apart from many older results that predate even parameterized complexity). In this framework, also called ``FPT in P'', one is interested in running times $\Oh(k^\alpha n^\beta)$ when the best dependence on the input size alone is $\Oh(n^\gamma)$ with $\gamma > \beta$, which then results in a better running time for sufficiently small parameter $k$. Typically, the parameter $k$ is at most $n$, thus, in the case of $\alpha + \beta = \gamma$, one already achieves truly better running times for $k \in o(n)$. We call such algorithms, which even for $k =n$ are not worse than the best unparameterized algorithm, \emph{adaptive} algorithms. 

Several recent publications dealt with efficient parameterized algorithms for different problems and parameters~\cite{FominLSPW18,IwataOO18,CoudertDP18,BentertFNN19,Husfeldt16,MertziosNN16}, however, \APSP got very little attention. Coudert et al.~\cite{CoudertDP18} considered the \emph{clique-width} of a graph as a parameter for tractable problems related to cycle problems. Intuitively, clique-width captures the closeness of a graph to a cograph, with cographs being exactly the graphs of clique-width at most two. Alongside some positive results for \prob{Triangle Counting} or \prob{Girth}, they proved a conditional lower bound for \prob{Diameter} namely that there is no $\Oh(2^{o(\cw)} \cdot n^{2 - \varepsilon})$ time algorithm for any $\varepsilon > 0$. That is, even computing just the greatest length of any shortest path in an unweighted graph admits no such algorithm. A weaker parameter and an upper bound for clique-width is the modular-width of a graph, which is another parameter that has been previously studied regarding its use for efficient parameterized algorithms~\cite{CoudertDP18,KratschN18}.

Note that small clique-width or small modular-width does not imply the sparsity of the graph, e.g. cliques have clique-width and modular-width two.
For parameters that do imply the sparsity of the graph (meaning that for parameter value $k$, the number of edges is bounded by $\Oh(kn)$, where $n$ denotes the number of vertices in the graph), the algorithm of Pettie and Ramachandran directly yields a running time of $\Oh(kn^2 + n^2 \log\log n)$, which is nearly optimal. 

\subparagraph{Our work.}
We study efficient parameterized algorithms for \APSP for its vertex-weighted variant. We consider the structural parameters \emph{clique-width} ($\cw$) and \emph{modular-width} ($\mw$).
As our main result, we present an $\Oh(\cw^2 n^2)$-time algorithm for \VWAPSP, yielding a truly subcubic algorithm for $\cw \in \Oh(n^{0.5 - \varepsilon})$. This immediately allows to solve the \prob{Diameter} problem in the same asymptotic time $\Oh(\cw^2 n^2)$, even with vertex weights, and thereby nicely complements the lower bound ruling out $\Oh(2^{o(\cw)} \cdot n^{2 - \varepsilon})$ for any $\varepsilon > 0$ \cite{CoudertDP18}.

Further, we present a general framework to determine the running time for many algorithms that use modular-width and the related modular decomposition tree. We use this framework to prove an algorithm of time
$\Oh(\mw^2n + n^2)$ for \prob{Vertex-Weighted All-Pairs Shortest Paths} on graphs of modular-width at most $\mw$.
This algorithm is combinatorial, however, it can benefit from subcubic algorithms for \prob{Vertex-Weighted All-Pairs Shortest Paths} that use fast matrix multiplication. For example, we achieve a running time of $\Oh(\mw^{1.842}n + n^2)$ by using an $\Oh(n^{2.842})$-time algorithm for the vertex-weighted case by Yuster~\cite{Yuster09} in each prime node; this algorithm uses fast matrix multiplication whereas all other algorithms (previous and new) are combinatorial. 
%

\subparagraph{Related Work.}
Following the work of Floyd and Warshall~\cite{Floyd62, Warshall62}, Fredman~\cite{Fredman76} achieved the first subcubic algorithm, running in time $\Oh(n^3 \log^{1/3}\log n / \log^{1/3}n)$. Chan~\cite{Chan10} and Han and Takaoka~\cite{HanT16} both achieved a running time of $\Oh(n^3 / \log^2 n)$ (omitting $poly(\log \log n)$ factors). Recently, Williams~\cite{Williams18} solved APSP in randomized time $\Oh(n^3 / 2^{\Omega(\log n)^{1/2}})$. For sparse graphs, Pettie and Ramachandran~\cite{PettieR05} get a running time of $\Oh( n^2 \alpha(n,m) + mn)$. 
All these algorithms solve the standard edge-weighted case.

In the vertex-weighted case, the currently fastest algorithm by Yuster~\cite{Yuster09} runs in $\Oh(n^{2.842})$ and relies on fast matrix multiplication. 
Shapira et al.~\cite{ShapiraYZ11} considered some variants of APSP, namely the \prob{All-Pairs Bottleneck Paths}, where one seeks the maximum bottleneck weight on a graph, and provided an algorithm of time $\Oh(n^{2.575})$ for vertex-weighted graphs. 
Czumaj and Lingas~\cite{CzumajL09} analyzed the related problem of finding the minimum-weight triangle in vertex-weighted graphs and achieved a running time of $\Oh(n^{\omega} + n^{2 + o(1)})$.
All of these algorithms for vertex-weighted graphs exploit fast matrix multiplication.
There is no truly subcubic \emph{combinatorial} algorithm known for \prob{Vertex-Weighted All-Pairs Shortest Paths}.

There are some subcubic algorithms known for APSP on special graph classes, such as uniform disk graphs with non-negative vertex weights, induced by point sets of bounded density within a unit square. Lingas and Sledneu~\cite{LingasS12} showed how to solve APSP on such graphs in time $\Oh(\sqrt{r} n^{2.75})$, where $r$ is the radius of the disk around the vertices in a unit square. 
Bentert and Nichterlein~\cite{BentertN19} considered the related problem of computing the diameter of a graph, parameterized by several parameters. 

\subparagraph{Organization.} Section~\ref{sec:Preliminaries} contains the preliminaries, in particular, the definition of clique-width and modular-width. In Section~\ref{sec:APSPcw}, we present the algorithm for \prob{Vertex-Weighted All-Pairs Shortest Paths} parameterized by the clique-width. The algorithm parameterized by modular-width as well as the running time framework can be found in Section~\ref{sec:APSPmw}. We conclude in Section~\ref{sec:Conclusion}.

\newpage
\section{Preliminaries}\label{sec:Preliminaries}
We follow basic graph notations~\cite{Diestel12}. For a natural number $k \in \mathbb{N}$, define $[k] = \{1, \ldots, k\}$. All graphs are simple, i.e., without loops or multiple edges. In a graph $G = (V, E)$, a path $P = (v_1, v_2, \ldots, v_n)$ is a sequence of vertices $v_i \in V$ with $\{v_i, v_{i+1}\} \in E$ for $i \in [n-1]$. We define by $P_{[v_i,v_j]}$ the subpath of $P$ starting in $v_i$ and ending in $v_j$ for $i, j \in [n]$ with $i < j$.  The \textit{length} of a path is the number of edges in it. In a vertex-weighted graph $G=(V,E)$ with weights $\omega \colon V \rightarrow \mathbb{R}_{\geq 0}$, the \textit{weight} (also called \textit{cost}) of a path $P$ is defined as $\omega(P) = \sum_{i = 1}^n \omega(v_i)$. Thus, every paths between two distinct vertices $u$ and $v$ has minimum weight $\omega(v) + \omega(u)$ and a path of length 0 from a vertex $v$ to itself has always weight $\omega(v)$. For a graph $G = (V,E)$ and $u,v \in V$, we denote the minimum weight of all paths between $u$ and $v$ as $dist_G(u,v)$. For a set of vertices $X \subseteq V$ and a vertex $u \in V$ we define $dist_G(u,X) = \min_{v \in X} dist_G(u,v)$ and for two sets of vertices $X,Y \subseteq V$, we define $dist_G(X,Y) = \min_{u \in X, v \in Y}(u,v)$. 

For two sets $A$ and $B$ we denote the disjoint union by $A \dot{\cup} B$ and we say that two sets $A$ and $B$ \emph{overlap} if $A \cap B \neq \emptyset$, $A \setminus B \neq \emptyset$, and $B \setminus A \neq \emptyset$.

\subsection{Clique-width and NLC-width}
A $k$-\textit{labeled} graph is a graph in which each vertex is assigned one out of $k$ labels. Formally, a vertex-labeled graph $G$ is a triple $(V,E,lab)$ with $V$ being the vertex set, $E$ denotes the set of edges, and $lab: V \rightarrow [k]$ is a function that defines the label for each vertex. For a $k$-labeled graph $G = (V,E,lab)$ we denote by $unlab(G) = (V,E)$ the underlying unlabeled graph. 
Intuitively, a graph $G$ has clique-width at most $k$, if it is the underlying graph of some $k$-labeled graph that can be constructed by using four operations: (1) Introducing a single labeled vertex, (2) redefining one label to another label, (3) taking the disjoint union of two already created $k$-labeled graphs, and (4) adding all edges between vertices of label $i$ to vertices of label $j$ for a pair $(i,j)$ of labels.

\begin{definition}[Clique-width,~\cite{CourcelleO00}]\label{def:cw}
 Let $k \geq 2$. The class $CW_k$ consists of all $k$-labeled graphs that can be constructed by the following operations:
 \begin{itemize} 
   	\item  The nullary operation $\bullet_a$, that corresponds to a graph consisting of a single                 vertex with a label $a \in [k]$. 
   	
   	\item Let $G = (V, E, lab) \in CW_k$ be a $k$-labeled graph, and let $a, b \in [k]$. Then 
	\begin{align*}
			\rho_{a,b}(G) = (V,E, lab') \qquad \text{ with } lab'(v) = 
			\begin{cases}
			  lab(v) &, \text{ if } lab(v) \neq a \\
			  b &, \text{ if } lab(v) = a \\
			\end{cases}	 
	\end{align*}
	is in $CW_k$.
        
        \item Let $G = (V_G, E_G, lab_G) \in CW_k$ and $H = (V_H, E_H, lab_H) \in CW_k$ be two $k$-labeled graphs in $CW_k$ with $V_G \cap V_H = \emptyset$. Then the disjoint union, defined by
  	\begin{align*}
  		 G \oplus H = (V_G \dot{\cup} V_H, E_G \dot{\cup} E_H, lab')  \qquad \text{ with } lab'(v) = 
  		 \begin{cases}
			lab_G(v) &, \text{ if } v \in V_G\\
			lab_H(v) &, \text{ if } v \in V_H\\
		\end{cases}
  	\end{align*}
	is in $CW_k$.
	
	\item Let $G = (V, E, lab ) \in CW_k$ be a $k$-labeled graph, and let $a, b \in [k]$ with $a \neq b$. Then 
	\begin{align*}
                \eta_{a,b}(G) = (V,E', lab) \qquad \text{ with } E' = E \cup \{ \{u,v\} \mid lab(u) = a, lab(v) = b \}
	\end{align*}
 \end{itemize}
\end{definition}

The clique-width of a graph $G$, denoted by $cw(G)$, is the smallest $k \geq 2$ such that there is a labeled graph $G' \in CW_k$ with $unlab(G') = G$.  The expression consisting of the operations defined in Definition~\ref{def:cw} is called a (clique-width) $k$-expression. For a $k$-expression $t$, we denote with $val(t)$ the resulting labeled graph and by $tree(t)$ the so called $k$-expression tree of $t$, which is the canonical tree representation of $t$.
Clique-width is a strict generalization of modular-width, which will be defined later. In fact, the clique-width of a graph $G$ is equal to the maximum clique-width of any quotient graph of a prime node in the modular decomposition tree of $G$. On the other hand, modular-width cannot be bounded by a function of clique-width.

Very similar to clique-width, one can define NLC-width, which was introduced by Wanke~\cite{Wanke94}. The main differences are that the join operation $\eta$ and the disjoint union operation $\oplus$ are somewhat combined and consecutive relabel operations are compressed into one operation.

\begin{definition}[NLC-width]\label{def:nlc}
 Let $k \geq 1$. The class $NLC_k$ consists of all $k$-labeled graphs that can be constructed by the following operations:
 \begin{itemize}
   	\item The nullary operation $\bullet_a$, that corresponds to a graph consisting of a single vertex with a label $a \in [k]$.
   	
   	\item Let $G = (V, E, lab) \in NLC_k$ and let $R \colon [k] \rightarrow [k]$. Then 
   	\begin{align*}
   		\circ_R(G) = (V,E,lab') \qquad \text{ with } lab'(v) = R(lab(v))
   	\end{align*}
	is in $NLC_k$.
	
	\item Let $G = (V_G, E_G, lab_G) \in NLC_k$ and $H = (V_H, E_H, lab_H) \in NLC_k$ be two $k$-labeled graphs in $NLC_k$. Let $S \subseteq [k]^2$. Then
	\begin{align*}
	 	&G \times_S H = (V_G \cup V_H, E', lab') \qquad \text{ with }  lab'(v) = 
  		 \begin{cases}
			lab_G(v) &, \text{ if } v \in V_G\\
			lab_H(v) &, \text{ if } v \in V_H\\
		\end{cases}\\
		&\text{ and } E' = E_G \cup E_H \cup \{\{u,v\} \mid u \in V_G, v \in V_H, \text{ and } (lab_G(u), lab_H(v)) \in S\}
	\end{align*}
	is in $NLC_k$.
 \end{itemize}

\end{definition}
The NLC-width of a graph $G$, denoted by $nlc(G)$, is the smallest $k \geq 2$ such that there is a labeled graph $G' \in NLC_k$ with $unlab(G') = G$. As for clique-width, the expression consisting of the operations defined in Definition~\ref{def:nlc} is called a (NLC-width) $k$-expression. For a $k$-expression $t$, we again denote with $val(t)$ the resulting labeled graph and by $tree(t) = T$ canonical tree representation of $t$, the so called $k$-expression tree of $t$.
This means each leaf node of $T$ is marked with $\bullet_a$ for some $a \in [k]$ and each internal node is either marked with $\circ_R$ for some $R \colon [k] \mapsto [k]$ or with $\times_S$ for some $S \subseteq [k]^2$, according to the operations defined in Definition~\ref{def:cw} resp.\ Definition~\ref{def:nlc}. For a node $x \in V(T)$ we denote by $G^x$ the labeled graph defined by the $k$-expression represented by the subtree of $T$ rooted in $x$ and we define by $L_i^x = \{v \in V(G^x) \mid lab(v) = i\}$ the set of vertex in $G^x$ with label $i \in [k]$. For a node $x \in V(T)$, we will use the shortcut $dist^x(u,v) := dist_{G^x}(u,v)$ to denote the distance between two vertices $u$ and $v$ in $G^x$.
 
The following lemma shows that we can safely focus on NLC $k$-expression trees, since the 
NLC-width and clique-width only differs by a factor of two at most.
\begin{lemma}[\cite{johansson98}]
 For any graph $G$ it holds that $nlc(G) \leq cw(G) \leq 2 \cdot nlc(G)$.
\end{lemma}

\subsection{Modular-width}

A \emph{module} in a graph $G = (V,E)$ is a set $M \subseteq V$ of vertices such that all vertices outside of $M$ are either connected to none or to all vertices in $M$, i.e, that $M \cap N(x) = \emptyset$ or $M \subseteq N(x)$ for every vertex $x \in V \setminus M$. Thus, all vertices in $M$ have the same neighborhood in 
$V\setminus M$.
It is easy to see that $\emptyset$, $V$, and $\{v\}$ for every $v \in V$ are modules of any graph $G = (V,E)$; those sets are called \emph{trivial modules}. If a graph $G$ only admits trivial modules, we call $G$ \emph{prime}.

Consider a so called modular partition $P = \{M_1, M_2, \ldots, M_\ell\}$ that is a partition of the vertices of $G$ into modules with $\ell \geq 2$. Due to the definition of a module, it holds for any two modules $M_i$ and $M_j$ of $P$ that every vertex of $M_i$ is either connected to all vertices of $M_j$ or to none. In the first case we call $M_i$ and $M_j$ \emph{adjacent}, in the latter \emph{non-adjacent}.
Thus, for a fixed modular partition $P$ we get a compact representation of the connection between the modules in $P$ by shrinking each module in $P$ to a single vertex. This graph is called the \emph{quotient graph} of $G$ (together with the modular partition $P$).

\begin{definition}
	Let $P = \{M_1, M_2, \ldots, M_\ell\}$ be a modular partition of a graph $G = (V, E)$. The \emph{quotient graph} is defined by $G_{/P} = (\{q_{M_1}, q_{M_2}, \ldots, q_{M_\ell}\}, E_P)$ with $E_P = \{ \{ q_{M_i}, q_{M_j} \} \mid \exists u \in M_i, v \in M_j \colon \{u,v\} \in E  \}$.
\end{definition}

For a modular partition $P$ of a graph $G$, the quotient graph $G_{/P}$ 
is a compact representation of all the edges in $G$ with endpoints in different modules.
If one additionally knows all subgraphs $G[M_i]$, with $i \in [\ell]$, one can reconstruct $G$. Each subgraph $G[M_i]$ is called a \emph{factor}. Instead of explicitly storing all factors, one can recursively decompose them as well until one reaches trivial modules $\{v\}$. To make the decomposition unique, one considers only modular partitions consisting of strong modules.
A module $M \subseteq V$ of a graph $G$ is called a \emph{strong} module, if it does not overlap with any other module $M'$ of $G$, meaning that either $M$ and $M'$ are disjoint or one module is a subset of the other. One can represent all strong modules of a graph $G$ by an inclusion tree $MD(G)$. 
Each strong module $M$ in $G$ corresponds to a vertex $v_M$ in $MD(G)$. 
A vertex $v_A$ is an an ancestor of $v_B$ in $MD(G)$ if and only if $B \subsetneq A$ for the corresponding strong modules $A$ and $B$ of $G$. Hence, the root node of $MD(G)$ corresponds always to the complete vertex set $V$ of $G$ and every leaf of $MD(G)$ corresponds a singleton set $\{v\}$ with $v \in V$.  Consider an internal node $v_M$ of $MD(G)$ with the set of children $\{v_{M_1}, \ldots, v_{M_\ell}\}$, i.e., $v_M$ corresponds to a strong module $M$ of $G$ and $P = \{M_1, \ldots , M_\ell\}$ is a modular partition of $G[M]$ into strong modules where $M_i$ is the corresponding module of $v_{M_i}$, with $i \in [\ell]$. There are three types of internal nodes in $MD(G)$. 
A node $v_M$ in $MD(G)$ is \emph{degenerate}, if for any non-empty subset of the children of $v_M$ in $MD(G)$, the union of the corresponding modules induces a (not necessarily strong) module. In this case the quotient graph $G[M]_{/P}$ is either a clique or an independent set. In the former case one calls $v_M$ a \emph{parallel} node, in the latter a \emph{series} node. Another case are so called \emph{prime} nodes. Here, for no proper subset of the children of $v_M$, the union of the corresponding modules induces a module. In this case the quotient graph of $v_M$ is prime.
 Gallai showed there are no further nodes in $MD(G)$.

\begin{theorem}[\cite{GallaiT67}] \label{modular_decomposition_theorem}
 For any graph $G = (V, E)$ one of the three conditions is satisfied:
 \begin{itemize}
    \item $G$ is not connected,
    \item $\overline{G}$ is not connected,
    \item $G$ and $\overline{G}$ are connected and the quotient graph $G_{/P}$, where $P$ is the maximal modular partition of $G$, is a prime graph.
 \end{itemize}

\end{theorem}
 
 Theorem~\ref{modular_decomposition_theorem} implies that $MD(G)$ is unique. The tree $MD(G)$ is called the \emph{modular decomposition tree} and the \emph{modular-width}, denoted by $\mw=\mw(G)$, is the minimum $k \geq 2$ such that any prime node in $MD(G)$ has at most $k$ children. Since every node in $MD(G)$ has at least two children and there are exactly $n$ leaves, $MD(G)$ has at most $2n -1$ nodes.
 It is known that $MD(G)$ can be computed in time $\Oh(n + m)$ \cite{TedderCHP08}. We refer to a survey of Habib and Paul~\cite{HabibP10} for more information.

\section{APSP parameterized by clique-width}\label{sec:APSPcw}

Assuming SETH, one cannot solve \prob{Diameter} (and thus, unweighted \APSP) in time $2^{o(\cw)} \cdot n^{2 - \varepsilon}$ \cite{CoudertDP18}. In this section, we show how to solve \VWAPSP in time $\Oh(\cw^2 n^2)$.

\begin{theorem}\label{thm:APSP}
    For every graph $G = (V, E)$, given together with a clique-width $k$-expression and vertex weights $\omega \colon V \rightarrow \mathbb{R}_{\geq 0}$, \prob{Vertex-Weighted All-Pairs Shortest Paths} can be solved in time $\Oh( k^2n^2)$. 
\end{theorem}

For an input graph $G=(V,E)$, given together with a clique-width $k$-expression for some $k \geq 2$, we transform in a first step the clique-width $k$-expression to an NLC-width $k$-expression in linear time as described for example in \cite{johansson98}. 
For the rest of this section, by writing $k$-expression we always refer to an NLC-width $k$-expression instead of a clique-width $k$-expression. We interpret the (NLC-width) $k$-expression as a $k$-expression tree $T$, in which each node $v \in V(T)$ is marked with an operation of the $k$-expression that is applied to the children of $v$. Accordingly, $T$ has exactly $n$ leafs, each marked with an operation $\bullet_i$ for $i \in [k]$, and exactly $n-1$ nodes marked with an operation $\times_S$ for some $S \subseteq [k]^2$. For ease of presentation, we assume that there is exactly one node marked with an operation $\circ_R$ for some $R \colon [k] \to [k]$ in between any two nodes marked with $\times_S$ (using $R(i)=i$ when no actual relabeling is necessary), hence, the length of the $k$-expression is $\Oh(n)$. Note, that the length of a clique-width $k$-expression is $\Oh(n+m)$ in general. For a node $x \in V(T)$ we denote by $G^x$ the labeled graph that is defined by the subexpression tree of $T$ rooted in $x$.

The algorithm consists of three phases. In the first phase, we traverse $T$ in a bottom-up manner: For each node $x \in V(T)$ we partition the vertex set into sets of same-labeled vertices and compute 
the shortest distance for each single vertex to (the closest vertex in) each label set.
Additionally, we compute the distance between each pair of label sets, i.e., the shortest distance of two vertices of the respective sets. Note, however, that in the first phase we only consider for each each node $x \in V(T)$ the distances in the graph $G^x$. In the second phase, we perform a top-down traversal of the $k$-expression tree $T$ and consider the whole graph $G$. 
Once we have computed the necessary values in phase one and two, we traverse $T$ one last time and finally compute the shortest path distances between all pairs of vertices.

\subparagraph{First Phase.} 

For a node $x \in V(T)$, which corresponds to the $k$-labeled graph $G^x$, we define $L_i^x = \{v \in V(G^x) \mid lab(v) = i\}$ as the set of all vertices in $G^x$ with label $i$. Note, that $unlab(G^x)$ is an induced subgraph of $G$ for any $x \in V(T)$.
We traverse $T$ in a bottom-up manner and compute for each node $x \in V(T)$ and for all pairs $(i,j) \in [k]^2$ of labels the shortest distance between some vertex in $L_i^x$ and some vertex in $L_j^x$. Additionally, we compute for any vertex $v \in V(G^x)$ and any label $i \in [k]$ the shortest distance from $v$ to some vertex in $L_i^x$. To be precise, for a node $x \in V(T)$ we compute the following values:
\begin{align*}
    c_{i,j}^x &= dist^x(L_i^x,L_j^x) &\text{ for } i,j \in [k] \\
    a_{v,i}^x &= dist^x(v, L_i^x)  &\text{ for } v \in V(G^x), i \in [k] 
\end{align*}

For nodes $x \in V(T)$ that are marked with $\times_S$ for some $S \subseteq [k]^2$ we need to compute some auxiliary values. Let $y$ and $z$ be the two children of $x$ in $T$. This means that $G^x$ consists of the disjoint union of $G^y$ and $G^z$ together with a full join between the vertex sets $L_i^y$ and $L_j^z$ for each $(i,j) \in S$. Thus, one can partition the vertex set of $G^x$ into the $2k$ sets $\{L_1^y, \ldots, L_k^y, L_1^z, \ldots, L_k^z\} = \mathcal{L}^x$. For each pair $(A,B) \in \mathcal{L}^x \times \mathcal{L}^x$ of vertex sets, we compute the shortest distance between some vertex in $A$ to some vertex of $B$. In addition, we compute the shortest distance between $A$ and $B$ with the constraint that either the first edge, the last edge, or the first and the last edge of the shortest path is an edge of a newly inserted full join defined by $S$. This achieves the effect that we additionally compute the shortest distance from (1) \textit{all} vertices of $A$ to \textit{some} vertex of $B$, (2) from \textit{some} vertex of $A$ to \textit{all} vertices of $B$, and (3) from \textit{all} vertices of $A$ to \textit{all} vertices of $B$. Doing this, one can e.g.\ combine a path that ends at \textit{some} vertex of $A$ with a path that can start at \textit{any} vertex of $A$.
In the following, we will describe how to compute the required values for each of the three different types of nodes in the $k$-expression tree $T$.

For the base case, let $x$ be a leaf of the $k$-expression tree $T$. Thus, the node $x \in V(T)$ is marked with $\bullet_\ell$ for some $\ell \in [k]$. This means that $G^x$ consists of a single vertex $v$ with label~$\ell$. In this case the following holds:
\begin{align*}
    c_{i,j}^x &= \begin{cases}
                    \omega(v) \qquad &\text{ if } i = j  = \ell \\
                    \infty \qquad &\text{ otherwise } \\
                \end{cases} 
                &\text{ for } i,j \in [k] \\
    a_{v,i}^x &= \begin{cases}
                    \omega(v) \qquad &\text{ if } i = \ell \\
                    \infty \qquad &\text{ otherwise } \\
                \end{cases}  
                &\text{ for } v \in V(G^x), i \in [k]  
\end{align*}
Now, let $x \in V(T)$ be an internal node of the $k$-expression tree $T$ marked with $\circ_R$ for some $R \colon [k] \to [k]$. Let $y \in V(T)$ be the unique child of $x$ in $T$. Since we traverse $T$ in a bottom-up manner, we have already computed the values $a_{v,i}^y$ for all $v \in V(G^y)$ and $i \in [k]$ and the values $c_{i,j}^x$ for all $i,j\in[k]$. Note, that $unlab(G^x) = unlab(G^y)$, which, in particular, implies that distances between vertices are identical in both graphs (though distances between label sets may be not, as these sets may be different).

\begin{lemma}
    Let $x \in V(T)$ be an internal node of a $k$-expression tree $T$ marked with $\circ_R$ for some $R \colon [k] \to [k]$ and let $y$ be the unique child of $x$ in $T$. Then 
    \begin{align}
        c_{i,j}^x = \min_{i' \in R^{-1}(i), j' \in R^{-1}(j)} c_{i',j'}^y \text{ for all } i,j \in [k].
    \end{align}

\end{lemma}

\begin{proof}
The vertex sets of $G^x$ and $G^y$ can be partitioned into the label sets $\{L_1^x, \ldots, L_k^x\}$ resp.\ $\{L_1^y, \ldots, L_k^y\}$ and it holds that $L_{i}^y \subseteq L_{R(i)}^x$ and $L_i^x = \bigcup_{j \in R^{-1}(i)}L_{j}^y$ for all $i \in [k]$. 
It follows that
\begin{align*}
c_{i,j}^x &= dist^x(L_i^x, L_j^x) \\
		  &= dist^y(\bigcup_{i' \in R^{-1}(i)}L_{i'}^y, \bigcup_{j' \in R^{-1}(j)}L_{j'}^y) \\
		  &= \min_{i' \in R^{-1}(i), j' \in R^{-1}(j)} dist^y(L_{i'}^y,L_{j'}^y) \\
		  &= \min_{i' \in R^{-1}(i), j' \in R^{-1}(j)} c_{i',j'}^y
\end{align*}
Here, it is crucial that distances between vertices are the same in $G^x$ and $G^y$, as noted above.
\end{proof}

Note, that the computation of all $c_{i,j}^x$ can be realized in time $\Oh(k^2)$ by updating for every $c_{i,j}^y$ the corresponding value $c_{R(i),R(j)}^x$. The values $a_{v,i}^x$ can be similarly computed from the values at the child node:

\begin{lemma}
    Let $x \in V(T)$ be an internal node of a $k$-expression tree $T$ marked with $\circ_R$ for some $R \colon [k] \to [k]$ and let $y$ be the unique child of $x$ in $T$. Then $a_{v,i}^x = \min_{j \in R^{-1}(i)} a_{v,j}^y$ for all $v \in V(G^x)$ and $i \in [k]$.
\end{lemma}

\begin{proof}
Let $\{L_1^x, \ldots, L_k^x\}$ resp.\ $\{L_1^y, \ldots, L_k^y\}$ be the partition in $G^x$ resp.\ $G^y$ of the vertex set into sets of same labeled vertices. Again, it holds that $L_{i}^y \subseteq L_{R(i)}^x$ and $L_i^x = \bigcup_{j \in R^{-1}(i)}L_{j}^y$ for all $i \in [k]$. Thus,
\begin{align*}
a_{v,i}^x &= dist^x(v,L_i^x) \\
		  &= dist^y(v, \bigcup_{j \in R^{-1}(i)}L_{j}^y) \\
		  &= \min_{j \in R^{-1}(i)} dist^y(v, L_j^y) \\
		  &= \min_{j \in R^{-1}(i)} a_{v,j}^y \qedhere
\end{align*} 
\end{proof}
The running time for computing the values $a_{v,i}^x$ is $\Oh(nk)$ since we need to consider each value $a_{v,i}^y$ for any $v \in V(G^y)$ and $i \in [k]$ exactly once.

Finally, let $x \in V(T)$ be an internal node of the $k$-expression tree $T$ marked with $\times_S$ for some $S \subseteq [k]^2$. Denote by $y \in V(T)$ and $z \in V(T)$ the two children of $x$ in $T$, meaning that $G^x$ combines the two labeled graphs $G^y$ and $G^z$ by introducing for each $(i,j) \in S$ a full join between the vertices in $L_i^y$ and those in $L_j^z$. Thus, $V(G^x) = V(G^y) \dot{\cup} V(G^z)$ and one can partition the vertices of $G^x$ into the $2k$ vertex sets $\{L_1^y, \ldots, L_k^y, L_1^z, \ldots, L_k^z\}$ and $L_i^x = L_i^y \dot{\cup} L_i^z$. 

To compute the desired distances between the label sets $\{L_1^y, \ldots, L_k^y, L_1^z, \ldots, L_k^z\}$, we construct an edge-weighted directed graph $H^x$ that represents all the distances between the label sets in a graph with only $4k$ vertices.

For each label set $L_i^a$ of $G^x$ with $i \in [k]$ and $a \in \{y,z\}$ we create two vertices $v_i^a$ and $u_i^a$. Let $V^a = \{v_i^a \mid i \in [k] \}$ resp.\ $U^a = \{u_i^a \mid i \in [k]\}$ for $a \in \{y,z\}$. We add a directed full join from $V^y$ to $U^y$ resp.\ from $V^z$ to $U^z$ with weight equal to the length of a shortest path between the two corresponding label sets. Finally, we connect vertices in $U^y$ with vertices in $V^z$, resp.\ vertices in $U^z$ with vertices in $V^y$, if and only if the corresponding pair is contained in $S$, i.e., if there is a full join in $G^x$ between the two corresponding label sets. See also Figure~\ref{Fig:Gstar} for an illustration. 
Formally, we define the directed, edge-weighted graph $H^x$ as follows.

\begin{figure}[t]

\centering

\begin{tikzpicture}

   \def \gap {3}
   \def \distance {3}
   \node[draw, circle] (Vy) at (0,0) {$V^y$};
   \node[draw, circle] (Uy) at (0,\distance) {$U^y$};
   \node[draw, circle] (Vz) at (\distance,\distance) {$V^z$};
   \node[draw, circle] (Uz) at (\distance,0) {$U^z$};
	
   \draw[->, >=latex] ([yshift=\gap] Vy.north) to node[left] {$c_{i,j}^y$} ([yshift=-\gap] Uy.south);
   \draw[->, >=latex] ([xshift=\gap] Uy.east) to node[above] {$S$} ([xshift=-\gap] Vz.west);
   \draw[->, >=latex] ([yshift=-\gap] Vz.south) to node[right] {$c_{i,j}^z$} ([yshift=\gap] Uz.north);
   \draw[->, >=latex] ([xshift=-\gap] Uz.west) to node[below] {$S^{-1}$} ([xshift=\gap] Vy.east);
    
\end{tikzpicture}
\caption{Construction of the auxiliary graph $H^x$. Each large node consists of $k$ disjoint vertices corresponding to one of the $k$ label sets in $G^y$ resp.\ $G^z$. Between $V^y$ and $U^y$ there is a full join, each edge between the corresponding vertex of $L_i^y$ and $L_j^y$ is weighted by $c_{i,j}^y$, analogously for $V^z$ and $U^z$. Vertices in $U^y$ are only connected to those vertices in $V^z$ for which the corresponding label sets are connected via $S$, analogously for $U^z$ and $V^y$.}
\label{Fig:Gstar}
\end{figure}
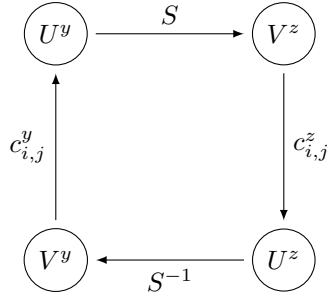

\begin{definition}\label{def:gstar}
Let $x \in V(T)$ be an internal node of a $k$-expression tree $T$ marked with $\times_S$ for some $S \subseteq [k]^2$ and let $y \in V(T)$ and $z \in V(T)$ be the children of $x$. We define $H^x$ as a directed, edge-weighted graph on $4k$ vertices created as follows: 
\begin{itemize}
	\item For each label set $L_i^a \in \{L_1^y, \ldots, L_k^y, L_1^z, \ldots, L_k^z\}$ we create two vertices $v_i^a$ and $u_i^a$ for $i \in [k]$ and $a \in \{y,z\}$. 
	\item Add edges $(v_i^y, u_j^y)$ with cost $c_{i,j}^y$ for all $i,j \in [k]$. 
	\item Add edges $(v_i^z, u_j^z)$ with cost $c_{i,j}^z$ for all $i,j \in [k]$.
	\item Add edges $(u_i^y, {v}_j^z)$ with cost zero for all $(i,j) \in S$.
	\item Add edges $(u_i^z, {v}_j^y)$ with cost zero for all $(j,i) \in S$.
\end{itemize} 
\end{definition}

Note, that some edges may have cost $\infty$ as there is no path of the requested type exists.
Next, we will see that $H^x$ exhibits all the desired distances from $G^x$ in a compact way.

\begin{theorem}\label{thm:gstarproperties}
Let $x \in V(T)$ be an internal node of a $k$-expression tree $T$ marked with $\times_S$ for some $S \subseteq [k]^2$ with children $y$ and $z$. Let $H^x$ be the graph as defined in Definition~\ref{def:gstar}. Then the following holds:
\begin{itemize}
    \item[(1)] $dist_{H^x}(v_i^a, u_j^b) = dist^x(L_i^a, L_j^b)$ for all $a, b  \in \{y,z\}$ and $i,j \in [k]$.
    
    \item[(2)] $dist_{H^x}(u_i^a, u_j^b) = \min_{P \in \mathcal{P}} \omega(P) - \min_{v \in L_i^a} \omega(v)$ where $\mathcal{P}$ is the set of all paths in $G^x$ starting in $L_i^a$, ending in $L_j^b$, and having the second vertex in $V(G^x) \setminus V(G[a])$.
    
    \item[(3)] $dist_{H^x}(v_i^a, v_j^b) = \min_{P \in \mathcal{P}} \omega(P) - \min_{v \in L_j^b} \omega(v)$ where $\mathcal{P}$ is the set of all paths in $G^x$ starting in $L_i^a$, ending in $L_j^b$, and having the penultimate vertex in $V(G^x) \setminus V(G[b])$.
    
    \item[(4)] $dist_{H^x}(u_i^a, {v}_j^b) = \min_{P \in \mathcal{P}} \omega(P) - \min_{v \in L_i^a} \omega(v) - \min_{v \in L_j^b} \omega(v)$ where $\mathcal{P}$ is the set of all paths in $G^x$ starting in $L_i^a$, ending in $L_j^b$, and having the second vertex in $V(G^x) \setminus V(G[a])$ and the penultimate vertex in $V(G^x) \setminus V(G[b])$.
\end{itemize}
\end{theorem}

We prove Theorem~\ref{thm:gstarproperties} in two steps. We first prove that every path in $H^x$ corresponds to some path in $G^x$. Later, we prove that also each optimal path between two label sets in $G^x$ corresponds to some shortest path in $H^x$. We start with statement (1) of Theorem~\ref{thm:gstarproperties}.

\begin{lemma}\label{lem:GstarToG}
 Let $x \in V(T)$ be an internal node of the $k$-expression tree $T$ marked with $\times_S$ for some $S \subseteq [k]^2$ with children $y$ and $z$. Let $H^x$ be the auxiliary graph as defined in Definition~\ref{def:gstar} and let $P^*$ be an arbitrary $v_i^y$-$u_j^z$-path in $H^x$ for some $i,j \in [k]$. Then there exists an $L_i^y$-$L_j^z$-path $P$ in $G^x$ with $\omega(P) = \omega_{H^x}(P^*) $. 
\end{lemma}

\begin{proof}
 Due to the circular structure of $H^x$, each $v_i^y$-$u_j^z$-path in $H^x$ will repeat the sequence $(v_p^y, u_q^y, v_r^z, u_s^z)$ for some $p,q,r,s \in [k]$ until reaching $u_j^z$ at the end of a sequence.
  Thus, each $v_i^y$-$u_j^z$-path in $H^x$ consists of $4\ell$ vertices for $\ell \in \mathbb{N}$ and can be written as
  \begin{align*}
 P^* = (v_i^y = v_{p_1}^y, u_{q_1}^y, v_{r_1}^z, u_{s_1}^z, v_{p_2}^y, u_{q_2}^y, v_{r_2}^z, u_{s_2}^z, \ldots , v_{p_\ell}^y, u_{q_\ell}^y, v_{r_\ell}^z, u_{s_\ell}^z = u_j^z).
  \end{align*}
  One can construct a path $P$ in $G^x$ from $P^*$ as follows: For each edge $(v_{p_i}^y, u_{q_i}^y)$ in $P^*$ of cost $c_{p_i,q_i}^y$ pick a shortest path in $G^y$ of total cost $c_{p_i,q_i}^y$ and for each edge $(v_{r_i}^z, u_{s_i}^z)$ in $P^*$ of cost $c_{r_i, s_i}^y$ pick a shortest path in $G^z$ of total cost $c_{r_i, s_i}^z$ for each $i \in [\ell]$. Those paths always exist since $c_{p_i,q_i}^y$ resp.\ $c_{r_i, s_i}^y$ are defined as the cost of a shortest $L_{p_i}^y$-$L_{q_i}^y$-path in $G^y$, resp.\ as the cost of a shortest $L_{r_i}^z$-$L_{s_i}^z$ in $G^z$. Since each edge $(u_{q_i}, v_{r_i})$ only exists if and only if there is a full join between the sets $L_{q_i}^y$ and $L_{r_i}^z$, one can connect the last vertex of the path corresponding to the previous edge in $P^*$ (that ends in some vertex in $L_{q_i}^y$) to the first vertex of the path corresponding to the following edge in $P^*$ (that starts at some vertex in $L_{r_i}^z$). 
 In the same manner one can argue that due to each edge $(u_{s_i}^z, v_{p_{i+1}}^y)$ one can connect the last vertex of the path corresponding to the edge $(v_{r_i}^z, u_{s_i}^z)$ with the first vertex of the path corresponding to the edge $(v_{p_{i+1}}^y, u_{q_{i+1}}^y)$. In both cases, the cost of the vertices is already accounted for in the resp.\ $c^\cdot_{\cdot,\cdot}$ value. Thus, each $v_i^y$-$u_j^z$-path in $H^x$ corresponds to an $L_i^y$-$L_j^z$-path in $G^x$ of same cost.
\end{proof}

Next, we generalize this argumentation to the following corollary:
\begin{corollary}\label{cor:gstarToGpath}
 Let $x \in V(T)$ be an internal node of the $k$-expression tree $T$ marked with $\times_S$ for some $S \subseteq [k]^2$ with children $y$ and $z$. Let $H^x$ be the auxiliary graph as defined in Definition~\ref{def:gstar}. 
 Then for every $i,j \in [k]$ and $a,b \in \{y,z\}$ the following holds:

 \begin{enumerate}
    \item[(1)]  For any $v_i^a$-$u_j^b$-path $P^*$ in $H^x$ there exists an $L_i^a$-$L_j^b$-path $P$ in $G^x$ with $\omega_{H^x}(P^*) = \omega(P)$.  	
    
    \item[(2)]  For any $u_i^a$-$u_j^b$-path $P^*$ in $H^x$ there exists an $L_i^a$-$L_j^b$-path $P = (p_1, p_2, \ldots, p_\ell)$ in $G^x$ with the property that $p_2 \in V(G^x)\setminus V(G[a])$ and $\omega_{H^x}(P^*) = \omega(P) - \omega(p_1)$. 
    
    \item[(3)]  For any $v_i^a$-$v_j^b$-path $P^*$ in $H^x$ there exists an $L_i^a$-$L_j^b$-path $P = (p_1, \ldots, p_{\ell-1}, p_\ell)$ in $G^x$ with the property that $p_{\ell-1} \in V(G^x)\setminus V(G[b])$ and $\omega_{H^x}(P^*) = \omega(P) - \omega(p_\ell)$. 
    
    \item[(4)]  For any $u_i^a$-$v_j^b$-path $P^*$ in $H^x$ there exists an $L_i^a$-$L_j^b$-path $P = (p_1, p_2, \ldots, p_{\ell-1}, p_\ell)$ in $G^x$ with the property that $p_2 \in V(G^x)\setminus V(G[a])$, $p_{\ell-1} \in V(G^x)\setminus V(G[b])$, and $\omega_{H^x}(P^*) = \omega(P) - \omega(p_1)- \omega(p_\ell)$. 
 \end{enumerate}
\end{corollary}

\begin{proof}
With the same argumentation as in the proof of Lemma~\ref{lem:GstarToG}, one can prove that also every $v_i^y$-$u_j^y$-path in $H^x$ corresponds to an $L_i^y$-$L_j^y$-path in $G^x$ with same cost. The two cases with $a=z$ and $b\in\{y,z\}$ now follow by swapping the roles of $y$ and $z$.
This proves (1).  

Similarly, one can argue that any path $P^*$ in $H^x$ that starts at a vertex $u_i^y$ (resp.\ ends at a vertex $v_i^z$) for $i \in [k]$, corresponds to a path $P = (p_1, p_2, \ldots, p_{\ell-1}, p_\ell)$ in $G^x$ that starts in $L_i^y$ (resp.\ ends in $L_j^z$) with the additional property (due to the directed edges) that $p_2 \in V(G^x) \setminus V(G^y)$ (resp.\ $p_{\ell-1} \in V(G^x) \setminus V(G^z)$). The cost of $P^*$ in $H^x$ is exactly the cost of $P$ in $G^x$, minus the cost of the first (resp.\ last) vertex of $P$. In this regard, recall that all vertex costs in $G^x$ are represented by the edge weights in $H^x$. The remaining cases of Corollary~\ref{cor:gstarToGpath} with $a = z$ or $b = y$ again can be shown analogously. 
\end{proof}

For any path in $H^x$ that starts at some vertex $u_i^a$ (resp.\ ends at some vertex $v_j^b$) one can find a corresponding path $P$ in $G^x$ with the property that the second vertex (resp.\ the penultimate vertex) is connected to \textit{all} vertices of $L_i^a$ (resp.\ $L_j^b$). Thus, one can extend any path that ends at \textit{some} vertex in $L_i^a$ by such a path (resp.\ one can prepend any path that starts in $L_j^b$ by such a path). Hence, the cost of the first vertex (resp.\ last vertex) is neglected if the path starts in some vertex $u_i^a$ or ends at some vertex $v_j^b$ for $a,b \in \{x,y\}$. 
In general, every path that one can find in $H^x$ corresponds to a path in $G^x$ of essentially the same cost, possibly without the first or last vertex (which can be chosen as the minimum of the label set). This proves ``$\leq$'' in the equations of Theorem~\ref{thm:gstarproperties}. 

For the other direction, we will show that every optimal shortest path between two label sets in $G^x$ is represented by a path in $H^x$.

\begin{lemma}\label{lem:optimumPathInGstar}
 Let $x \in V(T)$ be an internal node of a $k$-expression tree $T$ marked with $\times_S$ for some $S \subseteq [k]^2$ with children $y$ and $z$. Let $H^x$ be the auxiliary graph as defined in Definition~\ref{def:gstar}. 
 Let $P$ be a \textit{shortest} $L_i^y$-$L_j^z$-path in $G^x$ for some $i,j \in [k]$. Then there exist a $v_i^y$-$u_j^z$-path $P^*$ in $H^x$ with $\omega_{H^x}(P^*) = \omega(P)$. 
\end{lemma}

\begin{proof}
 Let $P$ be a shortest $L_i^y$-$L_j^z$-path in $G^x$ for fixed $i,j \in [k]$ of length $\omega(P) = dist^x(L_i^y, L_j^z)$.
 Since $V(G^x)$ can be subdivided into the two vertex sets $V(G^y)$ and $V(G^z)$, we can split the path $P$ into maximal subpaths consisting of vertices completely in $V(G^y)$ resp.\ completely in $V(G^z)$. Formally, let $P = (P_1, \ldots, P_{d})$ with $V(P_{2q+1}) \subseteq V(G^y)$ and $V(P_{2q}) \subseteq V(G^z)$ for $0 \leq q \leq d/2$. Let $v_i, u_i \in V(G)$ such that $P_i = P_{[v_i,u_i]}$ for all $i \in [d]$. Thus 
 \begin{align*}
  P = (\underbrace{v_1, \ldots, u_1}_{P_1 \in V(G^y)}, \underbrace{v_2, \ldots, u_2}_{P_2 \in V(G^z)}, \ldots ,\underbrace{v_{d-1}, \ldots, u_{d-1}}_{P_{d -1} \in V(G^y)}, \underbrace{v_d, \ldots, u_d}_{P_d \in V(G^z)}) 
 \end{align*}
 Note, that a path $P_\ell$ could consist of a single vertex $v_\ell=u_\ell$ for $\ell \in [d]$. 
 Define $E_{new}^x = E(G^x) \setminus (E(G^y) \cup E(G^z))$ as the set of newly created edges in $G^x$. I.e., $E_{new}^x$ is the union of all full joins between the vertex sets $L_i^y$ and $L_j^z$ for each $(i,j) \in S$.  
 For the path $P$ it holds by construction that $\{u_\ell, v_{\ell+1}\} \in E_{new}^x$ for $\ell \in [d-1]$ and, since $P$ is a shortest path, each $P_\ell$ is a shortest $v_\ell$-$u_\ell$-path in $G^y$ (for $\ell$ odd) resp.\ in $G^z$ (for $\ell$ even) for all $\ell \in [d]$. 
 
 To complete the proof, we will show that for each $\ell \in [d] $ it holds that $\omega(P_\ell) = c_{lab(v_\ell), lab(u_\ell)}^y$ for $\ell$ even and $\omega(P_\ell) = c_{lab(v_\ell), lab(u_\ell)}^z$ for $\ell$ odd, and thus that each path $P_\ell$ corresponds to an arc $(v_{lab(v_\ell)}^y, u_{lab(u_\ell)}^y)$ for $\ell$ even and to an arc $(v_{lab(v_\ell)}^z, u_{lab(u_\ell)}^z)$ for $\ell$ odd in $H^x$: 
 Let w.l.o.g.~$\ell$ be odd, hence $V(P_\ell) \subseteq V(G^y)$. Since $P_\ell$ is an $L_{lab(v_\ell)}^y$-$L_{lab(u_\ell)}^y$-path, it holds that $\omega(P_\ell) \geq c_{lab(v_\ell), lab(u_\ell)}^y$. Assume for contradiction that $\omega(P_\ell) > c_{lab(v_\ell), lab(u_\ell)}^y$ and let $Q$ be an $L_{lab(v_\ell)}^y$-$ L_{lab(u_\ell)}^y$-path in $G^y$ with $\omega(Q) = c_{lab(v_\ell), lab(u_\ell)}^y$. 
 Replace $P_\ell$ by $Q$ in $P$ to get $P' = (P_1, \ldots, P_{\ell-1}, Q, P_{\ell+1}, \ldots, P_d)$. Since $u_{\ell-1}$ is connected to all vertices in $L_{lab(v_\ell)}^y$, and $v_{\ell+1}$ is connected to all vertices in $L_{lab(u_\ell)}^y$, $P'$ is an $L_i^y$-$L_j^z$-path and since $\omega(P_\ell) > \omega(Q)$ the total cost of $P'$ is smaller than the total cost of $P$, which is a contradiction.    
\end{proof}

Again, one can generalize the argumentation of Lemma~\ref{lem:optimumPathInGstar} to the following corollary:
\begin{corollary}\label{cor:optimumPathInGstar}
 Let $x \in V(T)$ be an internal node of the $k$-expression tree $T$ marked with $\times_S$ for some $S \subseteq [k]^2$ with children $y$ and $z$. Let $H^x$ be the auxiliary graph as defined in Definition~\ref{def:gstar}. 
 Then for every $i, j \in [k]$ and $a,b \in \{y,z\}$ the following holds:
 \begin{itemize}
    \item[(1)] For every shortest $L_i^a$-$L_j^b$-path $P$ in $G^x$ there exists a $v_i^a$-$u_j^b$-path $P^*$ in $H^x$ of cost $\omega_{H^x}(P^*) = \omega(P)$.
    
    \item[(2)] For every shortest $L_i^a$-$L_j^b$-path $P$ in $G^x$ with the property that the second vertex is in $V(G^x) \setminus V(G[a])$ there exists a $u_i^a$-$u_j^b$-path $P^*$ in $H^x$ of cost $\omega_{H^x}(P^*) = \omega(P) - \min_{v \in L_i^a} \omega(v) $.
    
    \item[(3)] For every shortest $L_i^a$-$L_j^b$-path $P$ in $G^x$ with the property that the penultimate vertex is in $V(G^x) \setminus V(G[b])$ there exists a $v_i^a$-$v_j^b$-path $P^*$ in $H^x$ of cost $\omega_{H^x}(P^*) = \omega(P) - \min_{v \in L_j^b} \omega(v)$.    
    
    \item[(4)] For every shortest $L_i^a$-$L_j^b$-path $P$ in $G^x$ with the property that the second vertex is in $V(G^x) \setminus V(G[a])$ and the penultimate vertex is in $V(G^x) \setminus V(G[b])$ there exists a $u_i^a$-$v_j^b$-path in $H^x$ of cost $\omega_{H^x}(P^*) = \omega(P) - \min_{v \in L_i^a} \omega(v) - \min_{v \in L_j^b} \omega(v)$. 
 \end{itemize}
\end{corollary}

\begin{proof}
 With the same argumentation as done in Lemma~\ref{lem:optimumPathInGstar} one can show that also every shortest $L_i^y$-$L_j^y$-path in $G^x$ corresponds to a shortest $v_i^y$-$v_j^y$-path in $H^x$. The cases with $a = z$ follows analogously by renaming $y$ and $z$. This proves (1). 
 For the remaining cases we observe that for $i,j \in [k]$ and $a,b \in \{y,z\}$, every shortest $L_i^a$-$L_j^b$-path $P$ in $G^x$  with the property that the second vertex is in $V(G^x) \setminus V(G[a])$ (resp.\ the penultimate vertex in $V(G^x) \setminus V(G[b])$), the first (resp.\ last) vertex of $P$ is a minimum cost vertex in $L_i^a$ (resp.\ in $L_j^b$) and that it is not covered in the cost of the corresponding path in $H^x$.
\end{proof}

Corollary~\ref{cor:optimumPathInGstar} shows that every \textit{shortest} $L_i^a$-$L_j^b$-path in $G^x$ is represented in $H^x$ for $i,j \in [k]$ and $a,b \in \{y,z\}$. Together with Corollary~\ref{cor:gstarToGpath}, this proves Theorem~\ref{thm:gstarproperties}.

After the construction of the auxiliary graph $H^x$ as defined in Definition~\ref{def:gstar}, we compute and store the shortest distances for all pairs of vertices in $H^x$. With those values one can now compute the values $c_{i,j}^x$ and $a_{v,i}^x$ for $i,j \in [k]$ and $v \in V(G^x)$. Note that some of the values are only required in the second phase.

\begin{corollary}\label{cor:firstCij}
 Let $x \in V(T)$ be a node in the $k$-expression tree $T$ marked with $\times_S$ for some $S \subseteq [k]^2$. For all $i, j \in [k]$ it holds that
 \begin{align*}
  c_{i,j}^x = \min \big\{dist_{H^x}(v_i^y, u_j^y),dist_{H^x}(v_i^y, u_j^z),dist_{H^x}(v_i^z, u_j^y),dist_{H^x}(v_i^z, u_j^z) \big\}
\end{align*}
\end{corollary}

\begin{proof}
Since $V(G^x)$ is the disjoint union of $V(G^y)$ and $V(G^z)$ it holds that $L_i^x = L_i^y \dot{\cup} L_i^z$ for any $i \in [k]$. Thus, the shortest path between the vertex sets $L_i^x$ and $L_j^x$ for $i,j \in [k]$ in $G^x$ starts in either $L_i^y$ or $L_i^z$ and ends in either $L_j^y$ or $L_j^z$ and the following holds:  
\begin{align*}
	c_{i,j}^x &= dist^x(L_i^x, L_i^x) \\
                  &= dist^x(L_i^y \dot{\cup} L_i^z, L_j^y \dot{\cup} L_j^z) \\
                  &= \min \big\{dist^x(L_i^y,L_j^y),dist^x(L_i^y, L_j^z), dist^x(L_i^z, L_j^y),dist^x(L_i^z, L_j^z)\big\} \\
                  &= \min \big\{dist_{H^x}(v_i^y, u_j^y),dist_{H^x}(v_i^y, u_j^z),dist_{H^x}(v_i^z, u_j^y),dist_{H^x}(v_i^z, u_j^z) \big\}
\end{align*}
The last equation follows from Theorem~\ref{thm:gstarproperties}.
\end{proof}

\begin{corollary}\label{cor:firstAvi}
Let $x \in V(T)$ be a node in the $k$-expression tree $T$ marked with $\times_S$ for some $S \subseteq [k]^2$. Then for any $v \in V(G^y)$ and $i\in [k]$ it holds that
\begin{align*}
 a_{v,i}^x = \min_{j \in [k], a \in \{y,z\}} \big\{ a_{v,j}^y + dist_{H^x}(u_j^y, u_i^a)\big\}
\end{align*}
\end{corollary}

\begin{proof}
 Let $P = (v = w_1, w_2, \ldots, w_\ell)$ be a shortest $(v,L_i^x)$-path in $G^x$ with $v \in V(G^y)$. Subdivide $P$ into two parts $P = P_1,P_2$ such that $P_1$ is the maximal subpath of $P$ with $V(P_1) \subseteq V(G^y)$. Let $u_1$ be the last vertex of $P_1$. Then, $\omega(P_1) \geq a_{v,lab(u_1)}^y$ and $\omega(P_2) \geq dist_{H^x}(u_{lab(u_1)}^y, u_i^a) $, thus $a_{v,i}^x = \omega(P_1) + \omega(P_2) \geq \min_{j \in [k], a \in \{y,z\}} \{ a_{v,j}^y + dist_{H^x}(u_j^y, u_i^a)\}$. 
 
 On the other hand, for each $j \in [k]$ and $a \in \{y,z\}$ there is a path in $G^x$ of length $a_{v,j}^y + dist_{H^x}(u_j^y, u_i^a)$: Per definition there is a $(v, L_j^y)$-path $P_1$ in $G^x$ with $\omega(P_1) = a_{v,j}^y$ and due to Theorem~\ref{thm:gstarproperties}, there is a $L_j^y$-$L_i^z$-path $P_2$ of cost $\omega(P_2) = dist_{H^x}(u_j^y, u_i^a) - \min_{v \in L_j^y} \omega(v)$ with the property that the second vertex is connected to all vertices of $L_j^y$. Thus, one can combine the path $P_1$ with the path $P_2$ except of the first vertex of $P_2$. The resulting path is a $(v,L_i^a)$-path of desired cost.
\end{proof}

\subparagraph{Second Phase.}
In this phase, we process the $k$-expression tree $T$ in a top-down manner and use the local values that we have computed in the first phase to determine distances in the whole graph $G$. 

Consider an internal node $x \in V(T)$ of the $k$-expression tree $T$ marked with $\times_S$ for some $S \subseteq [k]^2$ and let $y$ and $z$ be the children of $x$ in $T$. Let $L_i^y$ resp.\ $L_i^z$ denote the set of vertices with label $i$ in $G^y$ resp.\ $G^z$ for $i \in [k]$. 
For an internal node $x$ with children $y$ and $z$ we will compute for any vertex set $L_i^y$ resp.\ $L_i^z$ and every vertex $v \in V(G^x)$ the minimum cost of all paths in $G$ that start in $v$ and end in $L_i^y$ resp.\ $L_i^z$ with the property that the penultimate vertex is in $V(G) \setminus V(G^y)$, resp.\ in $V(G) \setminus V(G^z)$. Thus, the penultimate vertex is connected to \textit{all} vertices of the vertex set $L_i^y$ resp.\ $L_i^z$. It will therefore be convenient not to include the cost of the final vertex in these costs (cf.\ definition below). Note, that we consider the whole graph $G$ in this step instead of just $G^x$.

Formally, for a node $x \in V(T)$ marked with $\times_S$ for some $S \subseteq [k]^2$ with children $y$ and $z$ we compute for every $v \in V(G^x)$ and $i \in [k]$ the following values:
\begin{itemize}
    \item $d_{v,i,y}^x = \min_{P \in \mathcal{P}} \omega(P) - \min_{u \in L_i^y} \omega(u)$ where $\mathcal{P}$ is the set of all paths in $G$ starting in $v$, ending in $L_i^y$, and having the penultimate vertex in $V(G) \setminus V(G^y)$.
    
    \item $d_{v,i,z}^x = \min_{P \in \mathcal{P}} \omega(P) - \min_{u \in L_i^z} \omega(u)$ where $\mathcal{P}$ is the set of all paths in $G$ starting in $v$, ending in $L_i^z$, and having the penultimate vertex in $V(G) \setminus V(G^z)$.
\end{itemize}

For a node $x \in V(T)$ marked with $\circ_R$ for some $R \colon [k] \mapsto [k]$ and the child $y$, we only compute $d_{v,i,y}^x$.
We start by computing those values for the root node. We can assume, w.l.o.g., that the root node has label $\times_S$ for some $S \subseteq [k]^2$.

\begin{lemma}\label{lem:dviyvaluesOnRoot}
 Let $r \in V(T)$ be the root node of the $k$-expression tree $T$ marked with $\times_S$ for some $S \subseteq [k]^2$ and let $y$ and $z$ be the children of $r$. Let further $H^r$ be the graph defined in Definition~\ref{def:gstar}. Then, for any $v \in V(G^y)$ and for every $i \in [k]$ it holds that
 \begin{align*}
 d_{v,i,y}^r = \min_{j \in [k]}\big\{a_{v,j}^y + dist_{H^r}(u_j^y, v_i^y)\big\},\\
 d_{v,i,z}^r = \min_{j \in [k]}\big\{a_{v,j}^y + dist_{H^r}(u_j^y, v_i^z)\big\}.
 \end{align*}
 Analogously, for any $v \in V(G^z)$ and for every $i \in [k]$ it holds that 
  \begin{align*}
 d_{v,i,y}^r = \min_{j \in [k]}\big\{a_{v,j}^z + dist_{H^r}(u_j^z, v_i^y)\big\},\\
 d_{v,i,z}^r = \min_{j \in [k]}\big\{a_{v,j}^z + dist_{H^r}(u_j^z, v_i^z)\big\}.
 \end{align*}
\end{lemma}

\begin{proof}
 We prove the first equation for the case that $v \in V(G^y)$. The second equation as well as both equations for $v \in V(G^z)$ can be shown analogously.

 We first show that $d_{v,i,y}^r \geq \min_{j \in [k]}\{a_{v,j}^y + dist_{H^r}(u_j^y, v_i^y)\}$.
 Consider a shortest path $P = (p_1, \ldots, p_{n-1}, p_n)$ with $p_1 = v$, $p_{n} \in L_i^y$ and $p_{n-1} \in V(G) \setminus V(G^y) = V(G^z)$. 
 Let $\ell$  be the maximal index such that 
 $p_r \in V(G^y)$ for all $1 \leq r \leq \ell$ and let $j \in [k]$ such that $p_{\ell} \in L_j^y$. It holds that $1\leq \ell \leq n-2$ since $p_1 \in V(G^y)$ and $p_{n-1} \notin V(G^y)$. Consider the subpaths $P_1 = (p_1, \ldots, p_{\ell})$ and $P_2 = (p_{\ell+1}, \ldots, p_n)$. By construction, $P_1$ is a $v$-$L_j^y$-path in $G^y$ and hence $\omega(P_1) \geq a_{v,j}^y$. 
 Since $p_{\ell+1}$ is connected to $p_\ell$ and $p_{\ell+1} \in G^z$, the vertex $p_{\ell+1}$ is connected to every vertex in $L_j^y$. We prepend $P_2$ by $p' = \argmin_{u \in L_j^y} \omega(u)$ and denote the resulting path by $P'_2$. Since $P$ is a shortest $v$-$L_i^y$-path and $p_{n-1} \in V(G^z)$, it holds that $p_n = \argmin_{u \in L_i^y} \omega(u)$. 
 The path $P'_2$ is an $L_j^y$-$L_i^y$-path in $G[r] = G$ with the property that the second vertex and the penultimate vertex are in $V(G[r]) \setminus V(G^y)$ and hence, due to Theorem~\ref{thm:gstarproperties}, $\omega(P'_2) \geq dist_{H^r}(u_j^y, v_i^y) + \omega(p') + \omega(p_n)$. Note, that $\omega(P) = \omega(P_1) + \omega(P'_2) - \omega(p')$. Thus 
 \begin{align*}
 	d_{v,i,y}^r &= \omega(P) - \min_{u \in L_i^y}  \omega(u) \\
 		    &= \omega(P_1) + \omega(P'_2) - \omega(p') - \min_{u \in L_i^y}  \omega(u) \\
 		    &\geq  a_{v,j}^y + dist_{H^r}(u_j^y, v_i^y) + \omega(p_n)- \min_{u \in L_i^y}\omega(u) \\
 		    &= a_{v,j}^y + dist_{H^r}(u_j^y, v_i^y)\\
 		    &\geq \min_{j \in [k]}\big\{a_{v,j}^y + dist_{H^r}(u_j^y, v_i^y)\big\} 
 \end{align*}
 For the other direction, we observe that for each $j \in [k]$ there is always a path in $G$ of cost $a_{v,j}^y + dist_{H^r}(u_j^y, v_i^y) + \min_{u \in L_i^y} \omega(u)$  that starts in $v$, ends in $L_i^y$, and that has its penultimate vertex in $V(G) \setminus V(G^y)$: For fixed $j \in [k]$ let $P_1$ be a shortest $v$-$L_j^y$-path in $G^y$. By definition it holds that $\omega(P_1) = a_{v,j}^y$. Let $P_2$ be a shortest $L_j^y$-$L_i^y$-path in $G[r] = G$ with the property that the second vertex and the penultimate vertex are in $V(G[r]) \setminus V(G^y)$. By Theorem~\ref{thm:gstarproperties} it holds that $\omega(P_2) = dist_{H^r}(u_j^y, v_i^y) + \min_{u \in L_j^y}\omega(u) + \min_{u \in L_i^y} \omega(u)$. Since the second vertex of $P_2$ is connected to \textit{all} vertices of $L_j^y$ we can combine $P_1$ and $P_2$  by removing the first vertex of $P_2$ to get a $v$-$L_i^y$-path $P$ in $G$ with the property that the penultimate vertex is in $V(G)\setminus V(G^y)$ of cost $a_{v,j}^y + dist_{H^r}(u_j^y, v_i^y) + \min_{u \in L_i^y} \omega(u)$. By definition of $d_{v,i,y}^r$ it follows that  $a_{v,j}^y + dist_{H^r}(u_j^y, v_i^y) \geq d_{v,i,y}^r$. Because the argument works for all $j\in[k]$, it follows directly that $d_{v,i,y}^r\leq \min_{j \in [k]}\{a_{v,j}^y + dist_{H^r}(u_j^y, v_i^y)\}$, which completes the proof.
\end{proof}

Next, we show how to propagate those values downwards in the $k$-expression tree, starting with a node marked with $\circ_R$ for some $R: [k] \mapsto [k]$.

\begin{lemma}
 Let $x \in V(T)$ be an internal node of the $k$-expression tree $T$ marked with $\circ_R$ for some $R \colon [k] \rightarrow [k]$. Let $y$ be the unique child of $x$ and $w$ be the unique ancestor of $x$ in $T$.  Then 
 $d_{v,i,y}^x = d_{v,R(i),x}^w$
\end{lemma}
\begin{proof}
 Since $x \in V(T)$ is marked with $\circ_R$ for some $R \colon [k] \rightarrow [k]$ it holds that $unlab(G^y) = unlab(G^x)$ and that $L_i^x = \bigcup_{j: R(j) = i} L_j^y$. Thus, all vertices in $V(G) \setminus V(G^x)$ are either connected to all vertices of $L_i^x$ or to none. Consider a shortest $v$-$L_i^y$-path $P$ in $G$ with the property that the penultimate vertex is in $V(G) \setminus V(G^y)$. Since $V(G^y) = V(G^x)$, $P$ is also a $v$-$L_{R(i)}^x$-path with the property that the penultimate vertex is in $V(G) \setminus V(G^x)$. Hence, $d_{v,i,y}^x \geq d_{v,R(i),x}^w$. On the other hand, every shortest $v$-$L_{R(i)}^x$-path in $G$ with the property that the penultimate vertex is in $V(G) \setminus V(G^x)$ can be changed to a $v$-$L_i^y$-path by possibly replacing the final vertex with $\argmin_{u \in L_i^y} \omega(u)$. Since the cost of this vertex is not included in $d$-values, it follows that $d_{v,R(i),x}^w \geq d_{v,i,y}^x$. 
\end{proof}

We now show the propagation for nodes $x$ of $T$ that are marked with $\times_S$. We start with one specific case and then conclude the general case as a corollary.

\begin{lemma}\label{lem:dviyvalueOnInternalNode}
 Let $x \in V(T)$ be an internal node of the $k$-expression tree $T$ that is marked with $\times_S$ for some $S \subseteq [k]^2$. Let $y$ and $z$ be the two children of $x$ in $T$, let $w$ be the unique ancestor of $x$ in $T$, and let $v \in V(G^y)$ be arbitrary. Then $d_{v,i,y}^x$ is the minimum of the following three values:
 \begin{itemize}
  \item $d_{v,i,x}^w$
  \item $\min_{j \in [k]} \big\{a_{v,j}^y + dist_{H^x}(u_j^y, v_i^y) \big\}$
  \item $\min_{j \in [k], c \in \{y,z\}} \big\{ d_{v,j,x}^w + dist_{H^x}(v_j^c, v_i^y)\big\}  $
 \end{itemize}

\end{lemma}

\begin{proof}
 After possibly adding nodes marked with $\circ_{id}$ to the $k$-expression tree, with $id$ being the identity function, one can assume, that $w$ is marked with $\circ_R$ for some $R \colon [k] \rightarrow [k]$ and that $x$ is the only child of $w$. 
 
 Let $P = (p_1, \ldots,p_{n-1}, p_n)$ be a shortest $v$-$L_i^x$-path in $G$ with penultimate vertex in $V(G) \setminus V(G^y)$, i.e., with $p_1 = v $, $p_n \in L_i^y$, and $p_{n-1} \in V(G) \setminus V(G^y)$; thus, $\omega(P) - \omega(p_n) = d_{v,i,y}^x $. We distinguish three cases: \\
 \textbf{Case 1:} $p_{n-1} \in V(G) \setminus V(G^x)$. In this case, $P$ is also a $v$-$L_i^x$-path with the property that the penultimate vertex is in $V(G) \setminus V(G^x)$; thus, $d_{v,i,y}^x \geq  d_{v,i,x}^w$. \\
 \textbf{Case 2:} $p_{n-1} \in V(G^z)$ and all vertices of $P$ are in $G^x$. In this case, we can compute the value in the same way as done in Lemma~\ref{lem:dviyvaluesOnRoot} for the root node and get $d_{v,i,y}^x = \min_{j \in [k]} \{a_{v,j}^y + dist_{H^x}(u_j^y, v_i^y) \}$. \\
 \textbf{Case 3:} $p_{n-1} \in V(G^z)$ and at least one vertex in $P$ is in $V(G) \setminus V(G^x)$. 
 Let $p_\ell$ be the last vertex of $P$ that is in $V(G) \setminus V(G^x)$; clearly, $p_{\ell+1}\in V(G^x)$. We split the path $P$ into the two subpaths $P_1 = (p_1, \ldots, p_{\ell})$ and $P_2 = (p_{\ell+1}, \ldots, p_n)$. Let $j \in [k]$ such that $p_{\ell+1} \in L_j^x$. 
 Since $p_\ell$ is connected to $p_{\ell+1}$, the vertex $p_\ell$ is connected to every vertex in $L_j^x$.  We extend $P_1$ by $p' = \argmin_{u \in L_j^x} \omega(u)$ and denote the resulting path by $P_1'$. 
 Now it holds by definition that $\omega(P_1') - \omega({p'}) \geq d_{v,j,x}^w$, as the penultimate vertex $p_\ell$ of $P'_1$ is in $V\setminus V(G^x)$.
 Let further $c \in \{y,z\}$ such that $p_{\ell+1} \in L_j^c$, noting that it does not change its label at $x$. Then $\omega(P_2) - \omega(p_{n}) \geq dist_{H^x}(v_j^c,v_i^y)$ by Theorem~\ref{thm:gstarproperties}, as $P_2$ is a path in $G^x$. Note, that $\omega(P) = \omega(P_1') + \omega(P_2) - \omega(p')$. Thus, in this case it holds that 
 \begin{align*}
  d_{v,i,y}^x &= \omega(P) - \min_{u \in L_i^y}  \omega(u) \\
  	      &= \omega(P'_1) - \omega(p') + \omega(P_2)  - \min_{u \in L_i^y}  \omega(u)\\
  	      &\geq d_{v,j,x}^w + dist_{H^x}(v_j^c,v_i^y) + \omega(p_n) - \min_{u \in L_i^y}  \omega(u)  \\
  	      &\geq d_{v,j,x}^w + dist_{H^x}(v_j^c,v_i^y) \\
  	      &\geq \min_{j \in [k], c \in \{y,z\}} \big\{ d_{v,j,x}^w + dist_{H^x}(v_j^c, v_i^y)\big\}
 \end{align*}
 We have seen in the case analysis above that in each case $d_{v,i,y}^x$ is at least the value considered in the case; in particular, it is at least equal to their collective minimum value. 
 On the other hand, for each case there is a path $P$ fulfilling the definition of $d_{v,i,y}^x$ such that $\omega(P) - \min_{u \in L_i^y} \omega(u)$ equals the value of the considered case. Thus, $d_{v,i,y}^x$ is also at most equal to the minimum taken over all three cases. This completes the proof.
\end{proof}

Lemma~\ref{lem:dviyvalueOnInternalNode} shows how to compute the value $d_{v,i,y}^x$ for any $v \in V(G^y)$. By a similar argumentation one can also compute the value $d_{v,i,\beta}^x$ for any $v \in V(G[\alpha])$ for $\alpha, \beta \in \{y,z\}$. 
 
\begin{corollary}
 Let $x \in V(T)$ be an internal node of the $k$-expression tree $T$ marked with $\times_S$ for some $S \subseteq [k]^2$. Let $y$ and $z$ be the unique children of $x$ in $T$, $w$ be the unique ancestor of $x$ in $T$, and let $\alpha, \beta \in \{y,z\}$ be arbitrary.
 Then, for $v \in V(G[\alpha])$ the value $d_{v,i,\beta}^x$ is the minimum of the following three values:
 \begin{itemize}
  \item $d_{v,i,x}^w$
  \item $\min_{j \in [k]} \big\{a_{v,j}^\alpha + dist_{H^x}(u_j^\alpha, v_i^\beta) \big\}$
  \item $\min_{j \in [k], c \in \{y,z\}} \big\{ d_{v,j,x}^w + dist_{H^x}(v_j^c, v_i^\beta)\big\}  $
 \end{itemize}
\end{corollary}

\subparagraph{Third Phase.}
In the third phase, we traverse the $k$-expression tree $T$ one final time; the ordering is immaterial. We go over all nodes $x$ with label $\times_S$ for some $S\subseteq [k]^2$ and compute for each pair of vertices $(u,v)$ with $u \in V(G^y)$ and $v \in V(G^z)$ the shortest $u$-$v$-path in $G$, where $y$ and $z$ are the two children of $x$ in $T$. 
Since the leaves of $T$ correspond one-to-one to single-vertex graphs, one for each vertex of $G$, this procedure will consider every pair of vertices in $G$ at some node $x \in V(T)$.

\begin{lemma}\label{lem:ThirdTraversal}
 Let $x \in V(T)$ be an internal node of the $k$-expression tree $T$ marked with $\times_S$ for some $S\subseteq [k]^2$. Let $y$ and $z$ be the two children of $x$ and let $u \in V(G^y)$ and $v \in V(G^z)$. Then $dist_G(u,v) = \min_{i \in [k]} \big\{d_{u,i,z}^x + a_{v,i}^z \big\}$.
\end{lemma}

\begin{proof}
 Let $P = (u = p_1, \ldots, p_n = v)$ be a shortest $u$-$v$-path in $G$. Let $\ell$ be the largest index such that $p_\ell \in V(G) \setminus V(G^z)$. Since $p_1 \in V(G^y)$ and $p_n \in V(G^z)$ this index must exist and it holds that $\ell \leq n-1$. Split $P$ into two subpaths $P_1 = (p_1, \ldots, p_\ell)$ and $P_2 = (p_{\ell+1}, \ldots, p_n)$. 
 Let $i \in [k]$, such that $p_{\ell+1} \in L_i^z$. 
 Since $p_\ell$ is connected to $p_{\ell+1}$, the vertex $p_\ell$ is connected to every vertex in $L_i^z$. We extend $P_1$ by $p' = \argmin_{u \in L_i^z} \omega(u)$ and denote the resulting path by $P_1'$. Now, $P'_1$ is a $u$-$L_i^z$-path with penultimate vertex in $V(G)\setminus V(G^z)$ and, therefore $\omega(P_1') - \omega(p') \geq d_{u,i,z}^x$. Similarly, using that $G$ is undirected, the reverse of $P_2$ is a $v$-$L_i^z$-path in $G^z$, implying that $\omega(P_2) \geq a_{v,i}^z$. Thus 
 \begin{align*}
 dist_G(u,v) &= \omega(P) \\
             &= \omega(P_1') - \omega(p') + \omega(P_2) \\
             &\geq d_{u,i,z}^x + a_{v,i}^z \\
             &\geq \min_{i \in [k]} \{d^x_{u,i,z} + a_{v,i}^z \}. 
 \end{align*}
 
 Conversely, we show that for each $i \in [k]$ there is a $u$-$v$-path in $G$ of cost $d^x_{u,i,z} + a_{v,i}^z$: For fixed $i \in [k]$ let $P'$ be a shortest $u$-$L_i^z$-path in $G$ with the property that the penultimate vertex is in $V(G) \setminus V(G^z)$. By definition it holds that $\omega(P') = d^x_{u,i,z} + \min_{w \in L_i^z} \omega(w)$. Let $P_1$ be obtained from $P'$ by removing the last vertex. Now, $\omega(P_1) = d_{u,i,z}^x$ and $P_1$ has the property that it starts in $u$ and that its last vertex is adjacent to all vertices of $L_i^z$. Let $P_2$ be a shortest $v$-$L_i^z$-path in $G^z$. By definition it holds that $\omega(P_2) = a_{v,i}^z$. Now, we can extend $P_1$ by the reverse of $P_2$ to get a $u$-$v$-path in $G$ of cost $d^x_{u,i,z} + a_{v,i}^z$. This implies that $dist_G(u,v)\leq d^x_{u,i,z} + a_{v,i}^z$ and, using that the argument works for all $i\in[k]$, that $dist_G(u,v) \leq \min_{i \in [k]} \big\{d_{u,i,z}^x + a_{v,i}^z \big\}$.
 \end{proof}

\subparagraph{Running time.}

First, we need to transform the clique-width $k$-expression into a NLC-width $k$-expression tree $T$, which can be done in linear time $\Oh(n+m)$ \cite{johansson98}. 

In the first traversal, we compute for every node $x \in V(T)$ the values $a_{v,i}^x$ for $v \in V(G^x)$ and $i \in [k]$. Thus, we compute at most $n \cdot k$ values, each in time $\Oh(k)$, which results in a running time of $\Oh(n k^2)$ per node of $T$.
In the case of a node $x$ with label $\times_S$ for some $S \subseteq [k]^2$ we first compute the auxiliary graph $H^x$ in time $\Oh(\vert V(H) \vert + \vert E(H) \vert) = \Oh(k^2)$ and solve (edge-weighted) \APSP on $H^x$ in time $\Oh(k^3)$. After this, by using Corollary~\ref{cor:firstCij} resp.\ Corollary~\ref{cor:firstAvi}, we compute each $c_{i,j}^x$ in constant time resp.\ each $a_{v,i}^x$ in time $\Oh(k)$ resulting in a running time per node $x \in V(T)$ of $\Oh(k^3 + k^2 \cdot n)$.

In the second phase we perform a top-down traversal of $T$ to compute the for each node $x$ the values $d_{v,i,y}^x$ and $d_{v,i,z}^x$ for all $v \in G^x$ and $i \in [k]$. Again, we compute at most $n \cdot k$ values, each in time $\Oh(k)$, which results in a running time of $\Oh(nk^2)$ per node of $T$. 
Since there are $\Oh(n)$ nodes in the $k$-expression tree $T$, the total running time for Phase One and Phase Two is $\Oh(nk^3 + n^2k^2) = \Oh(n^2k^2)$.

In the last phase, we consider each pair $(u,v)$ of vertices exactly once and compute each pairwise distance in time $\Oh(k)$. Thus, running time for the last phase is $\Oh(n^2k)$. In total, we obtain the claimed bound of $\Oh(k^2n^2)$.

\section{APSP parameterized by modular-width}\label{sec:APSPmw}

\subsection{General Running Time Theorem}
First, we will derive a general running time theorem that is applicable to many algorithms that use modular decomposition trees, for example all algorithms in \cite{KratschN18}. 
Since we will focus on functions describing running times,
we will restrict ourselves to functions $T \colon \mathbb{R}_{\geq 1} \rightarrow \mathbb{R}_{\geq 1}$.

\begin{definition}[\cite{BuraiSA05}]
 A function $T \colon \mathbb{R}_{\geq 1} \rightarrow \mathbb{R}_{\geq 1}$ is \emph{superhomogeneous} if for all $ \lambda \geq 1$ the following holds:
 \begin{align*}
  \lambda \cdot T(n) \leq T(\lambda \cdot n)
 \end{align*}
\end{definition}

\begin{lemma}\label{lem:properySuperhomogeneous}
 Let $T \colon \mathbb{R}_{\geq 1}^2 \rightarrow \mathbb{R}_{\geq 1}$ be a function that is superhomogeneous in the first component and  monotonically increasing in the second component. Then
 \begin{align*}
 \max_{ \substack{1 \leq k \leq n \\ 1 \leq \ell \leq m}} \frac{T(k,\ell)}{k} \leq \frac{T(n,m)}{n}.
 \end{align*}
\end{lemma}

\begin{proof}
 Since $T$ is monotonically increasing in the second component, the maximum is reached for $\ell = m$. Pick $k$ with $1\leq k\leq n$ arbitrarily and set $\lambda\geq 1$ so that $\lambda\cdot k=n$. It follows directly that
 \[
  \frac{T(n, m)}{n} = \frac{T(\lambda k,m)}{\lambda k} \geq \frac{\lambda T(k, m)}{\lambda k} = \frac{T(k,m)}{k}.
 \]
 This completes the proof.
\end{proof}

We can now state the running time framework.

\begin{theorem}\label{thm:RunningTimeMDT}
 Let $G$ be a graph of modular-width equal to $\mw$, let $MD(G)$ be the modular decomposition tree of $G$, and let $T \colon \mathbb{R}_{\geq 1}^2 \rightarrow \mathbb{R}_{\geq 1}$ be a function that is superhomogeneous in the first component and monotone increasing in the second component. If the running time of an algorithm for any prime node $v_i \in V(MD(G))$ is upper bounded by $\Oh(T(n_i, m_i))$, where $n_i$ and $m_i$ denote the number of vertices and edges of the quotient graph corresponding to $v_i$, then the total running time can be upper bounded by 
 \begin{align*}
  \Oh\left(\frac{n}{\mw} \cdot T(\mw, m) + n + m \right)\qquad \text{ and }  \qquad\Oh\left(\frac{n}{\mw} \cdot  T(\mw, \mw^2) + n + m\right).
 \end{align*}
If, additionally, $T$ is also superhomogeneous in the second component then the running time can also be upper bounded by $\Oh(T(\mw, m) + n + m)$.
\end{theorem}

\begin{proof}
 In a first step we compute the modular decomposition tree of the input graph $G$ in linear time~\cite{TedderCHP08}.
 We can assume that the running time is then dominated by the sum of computations of all prime nodes, since one can replace a series resp.\ parallel node with $k$ children by a sequence of $k-1$ pseudo-prime nodes each with a corresponding quotient graph isomorphic to $K_2$ resp.\ $I_2$. Let $t$ denote the number of prime nodes after this replacement. For any node $v_i \in V(MD(G))$ let $n_i$ and $m_i$ denote the number of vertices resp.\ edges in the quotient graph associated with $v_i$. Thus, it holds that $n_i \leq \mw $ and $m_i \leq \mw^2$. Therefore, the running time all nodes can be upper bounded by
 \begin{align}
 \sum_{i = 1}^t T(n_i,m_i)  &=     \sum_{i = 1}^t n_i \frac{T(n_i,m_i)}{n_i} \notag \\
 		            &\leq  \sum_{i = 1}^t n_i \cdot \left( \max_{ \substack{1 \leq n_i \leq \mw \\ 1 \leq m_i \leq m}} \frac{T(n_i,m_i)}{n_i} \right) \label{eq:mi} \\
 		            &\leq 2n \cdot \frac{T(\mw,m)}{\mw}. \notag
\end{align}
The last inequality holds due to Lemma~\ref{lem:properySuperhomogeneous} and since $\sum_{i=1}^t n_i$ counts each node in the modular decomposition (except of the root) exactly once. 

Since $T$ is monotone increasing in the second component and each quotient graph has at most $\mw^2$ many edges, one can replace $m_i$ by $\mw^2$ in (\ref{eq:mi}) and get a running time for processing all nodes in the modular decomposition tree of $\Oh(n \frac{T(\mw, \mw^2)}{\mw})$.

If, additionally, $T$ is also superhomogeneous in the second component then the running time can be upper bounded by
\begin{align*}
 \sum_{i = 1}^t T(n_i,m_i)  &=  \sum_{i = 1}^t m_i \frac{T(n_i,m_i)}{m_i}\\
                            &\leq \sum_{i = 1}^t m_i \cdot \left( \max_{ \substack{1 \leq n_i \leq \mw \\ 1 \leq m_i \leq m}} \frac{T(n_i,m_i)}{m_i} \right) \\
                             &\leq m \cdot \frac{T(\mw,m)}{m}\\
                             &= T(\mw,m).
\end{align*}
The second to last inequality holds using the argumentation from Lemma~\ref{lem:properySuperhomogeneous} and since $\sum_{i = 1}^t m_i$ counts every edge at most once.
\end{proof}

\subparagraph{Example.} In \cite{KratschN18} it was shown how to solve \prob{Global Minimum Vertex Cut} with a running time per prime node of $\Oh(n_i m_i \log n_i)$, where $n_i$ denotes the number of vertices and $m_i$ denotes the number of edges in the quotient graph of prime node $v_{M_i}$ in the modular decomposition tree. Thus, by using Theorem~\ref{thm:RunningTimeMDT}, one can bound the total running time by $\Oh(\min\{n \mw^2 \log \mw, m \mw \log \mw\} + n + m)$.

\subsection{APSP parameterized by modular-width}

In this section we study \VWAPSP relative to modular-width. We obtain the following result.

\begin{theorem}\label{thm:APSPmw}
 For every graph $G = (V,E)$ of modular-width $\mw$ and with given vertex weights $\omega \colon V \rightarrow \mathbb{R}_{\geq 0}$, \prob{Vertex-Weighted All-Pairs Shortest Paths} can be solved in time $\Oh(\mw^{1.842}n + n^2)$ using fast matrix multiplication and otherwise in time $\Oh(\mw^2n + n^2)$. 
\end{theorem}

Let $MD(G)$ be the modular decomposition tree of $G$, which can be computed in linear time.
We will traverse $MD(G)$ in a top-down manner.
For a node $v_M$ in the decomposition tree with children $v_{M_1}, \ldots, v_{M_\ell}$ let $M$ resp.\ $M_1, \ldots, M_\ell$ be the corresponding modules in $G$. We will compute for every $u,v \in M$ with $u \in M_i$ and $v \in M_j$ with $i \neq j$ the shortest path in the whole graph $G$. Since every vertex of $G$ corresponds to a leaf in $MD(G)$, we eventually consider every pair of vertices with this procedure. We start with some structural properties of shortest paths in the modular decomposition tree.

\begin{lemma}\label{lem:disjoint_module_path}
Let $P = \{M_1, \ldots, M_\ell\}$ be a modular partition of a graph $G = (V,E)$. Let $s,t \in V$ be two vertices with $\{s,t\} \nsubseteq M_i$ for all $i \in [\ell]$. Then, there exists a shortest $s$-$t$-path that visits each module at most once.
\end{lemma}

\begin{proof}
 Let $Q = (s = v_1, v_2, \ldots, v_n = t)$ be a shortest $s$-$t$-path in $G$. Assume that there exist $i,j \in [n]$ with $i \neq j$ such that $\{v_i, v_j\} \subseteq M_k$ for some $k \in [\ell]$. Let $i$ be minimal and $j$ be maximal under this condition. We distinguish two cases: In the case $j \neq n$, consider the path $Q' = (v_1, \ldots, v_{i-1}, v_i, v_{j+1}, v_{j+2}, \ldots, v_n)$. Since $j$ is maximal and $v_j \neq v_n$, it holds that $v_{j+1} \notin M_k$, but $v_{j+1}$ is adjacent to all vertices of $M_k$. Thus, the edge $\{v_i,v_{j+1}\}$ exists and $Q'$ is indeed a $s$-$t$-path with $\omega(Q') \leq \omega(Q) = dist_G(s,t)$.
 
 If $v_j = v_n$ then it holds that $v_i \neq v_1$ since otherwise $\{s,t\}$ are in a same module. Consider the path $Q'' = (v_1, \ldots, v_{i-1}, v_j)$. Since $i$ is minimal, it holds that $v_{i-1} \notin M_k$, but $v_{i-1}$ is adjacent to all vertices of $M_k$, in particular to $v_j$. Thus, $Q''$ is an $s$-$t$-path with $\omega(Q'') \leq \omega(Q) = dist_G(s,t)$. 
We iterate this procedure for every pair of vertices that are in a same module. Since the vertices in $Q'$ resp. $Q''$ are a strict subset of the vertices in $Q$, the number of pairs that are in a same module strictly reduces each time.
\end{proof}

We will use this property to compute shortest paths between vertices in different modules. To do so, we extend the quotient graph  by vertex weights.

\begin{definition}\label{def:weightedQuotientGraph}
 Let $G = (V,E)$ be a graph and $P = \{M_1, \ldots, M_\ell\}$ be a modular partition of $G$. Define $G_{/P}^*$ as the quotient graph $G_{/P}$ extended by additional vertex weights $\omega_{G_{/P}^*}(q_{M_i}) = \min_{v \in M_i} \omega_G(v)$.
\end{definition}

\begin{lemma}\label{lem:computeDistDisjointVertices}
 Let $G = (V,E)$ be a graph and let $P = \{M_1, \ldots, M_\ell\}$ be a modular partition of $G$. Let $G_{/P}^*$ be the vertex-weighted quotient graph as defined in Definition~\ref{def:weightedQuotientGraph} and let $u,v \in G$ be two vertices with $u \in M_i$ and $v \in M_j$ with $i \neq j$. Then, $dist_G(u,v) = dist_{G_{/P}^*}(q_{M_i}, q_{M_j}) - \omega_{G_{/P}^*}(q_{M_i}) + \omega_G(u) - \omega_{G_{/P}^*}(q_{M_j}) + \omega_G(v) $.
\end{lemma}

\begin{proof}
Let $u,v \in G$ be two vertices in $G$ with $u \in M_i$ and $v \in M_j$ with $i \neq j$.
Every shortest $q_{M_i}$-$q_{M_j}$-path $P^*$ in $G_{/P}^*$ corresponds to a $u$-$v$-path $P$ in $G$ by first replacing each vertex in $P^*$ by the minimum-weight vertex of the corresponding module, and afterwards, since $M_i$ and $M_j$ are modules, replacing the first vertex by $u$ and the last vertex by $v$.

Conversely, let $P$ be a shortest $u$-$v$-path in $G$ with $u \in M_i$ and $v \in M_j$ with $i \neq j$. Due to Lemma~\ref{lem:disjoint_module_path}, we can assume that no two vertices of $P$ are in the same module. Thus, due to the structure of modules, one can assume that each vertex of $P$, except of $u$ and $v$, are minimum weight vertices of their respective module. Hence, there is a $q_{M_i}$-$q_{M_j}$-path in $G_{/P}^*$ of cost $dist_G(u,v) - \omega_G(v) + \omega_{G_{/P}^*}(q_{M_i}) - \omega_G(u) + \omega_{G_{/P}^*}(q_{M_j})$, which proves the claim. 
\end{proof}

Due to Lemma~\ref{lem:computeDistDisjointVertices}, one can compute the shortest path length for vertices in different modules by solving the \prob{Vertex-Weighted All-Pairs Shortest Paths} problem on $G_{/P}^*$. The next lemma shows that for vertices that are in a same module, either the entire shortest path between them is inside this module or it is a path of length two.

\begin{corollary}\label{cor:eitherInsideOrHop}
Let $G = (V,E)$ be a graph with non-negative vertex weights $\omega_G \colon V \rightarrow \mathbb{R}_{\geq 0}$. Let $M \subseteq V$ be a module in $G$ and $P = \{M_1, \ldots, M_\ell\}$ be a modular partition of $G[M]$. Let further $u,v \in M$ be two vertices with $\{u,v\} \subseteq M_i$ for an $i \in [\ell]$.
Then, every shortest $u$-$v$-path in $G$ is either completely inside $G[M]$ or there exists a shortest $u$-$v$-path with exactly two edges.
\end{corollary}

\begin{proof}
Assume that there is a shortest $u$-$v$-path $P$ in $G$ that is not completely inside $G[M]$ and contains more than two edges; in other words, $P$ contains at least two vertices $p,q\in V\setminus M$. Since $M$ is a module in $G$ and $u, v \in M$, every vertex $x \in V \setminus M$ is either connected to both $u$ and $v$ or to neither $u$ nor $v$. Thus, one can shortcut $P = (u, \ldots, p,\ldots,q,\ldots,v)$ to $P' = (u,\ldots,p,\ldots, v)$ and since every vertex weight is non-negative, it holds that $\omega(P') \leq \omega(P)$.
\end{proof}

We define the slightly more general problem \prob{$k$-Capped Vertex-Weighted APSP}, that takes a vertex-weighted graph $G$ as an input and asks for all pairs of vertices $u,v$ for the value $d_k(u,v) = \min \big\{ dist_G(u,v), k \big\}$. This generalizes \VWAPSP if we set $k$ large enough, i.e., set $k = \sum_{v \in V} \omega(v)$. We can now describe the algorithm and prove Theorem~\ref{thm:APSPmw}:

\begin{proof}[Proof of Theorem~\ref{thm:APSPmw}] 
For an input graph $G = (V,E)$, the algorithm first computes the modular decomposition tree $MD(G)$ and then processes $MD(G)$ in a top-down traversal, starting with the root node. For a node $v_M$ in $MD(G)$ with children $v_{M_1}, \ldots, v_{M_\ell}$, let $M$ be the corresponding module and $P = \{M_1, \ldots, M_\ell\}$ be the corresponding modular partition of $G[M]$. 
We solve \prob{$k$-Capped APSP} in $G[M]$ as follows: 

First, we construct the weighted quotient graph $G_{/P}^*$ as defined in Definition~\ref{def:weightedQuotientGraph}, and solve vertex-weighted APSP on $G_{/P}^*$.  
Next, we compute for all pairs of vertices in $G[M]$ that are in different modules $M_i$ the shortest distance in $G$ by utilizing Lemma~\ref{lem:computeDistDisjointVertices}. 
Afterwards, we compute for each module $M_i \in P$ the minimum weight of all vertices in adjacent modules using $G_{/P}^*$, i.e, $k_i = \min_{q_{M_j} \in N(q_{M_i})} \omega(q_{M_j})$. Finally, we use Corollary~\ref{cor:eitherInsideOrHop} and recurse by solving for each module $M_i \in P$ the Problem \prob{$k$-Capped APSP} on $G[M_i]$ with $k = k_i$. 
For the root node, we set $k = \sum_{v \in V} \omega(v)$.

For any prime node $v_{M_i}$ in $MD(G)$ vertex-weighted APSP can be solved in time $\Oh(n_i^{2.842})$ with an algorithm due Yuster~\cite{Yuster09}, where $n_i$ denotes the number of vertices in the corresponding quotient graph. With a standard combinatorial algorithm one can solve vertex-weighted APSP in time $\Oh(n_i^3)$.
Thus, by Theorem~\ref{thm:RunningTimeMDT}, the total running time for this step is $\Oh(n\mw^{1.842} + m)$ if we use fast matrix multiplication or $\Oh(n\mw^2 + m)$ for the combinatorial algorithm. 
After we have solved vertex-weighted APSP on all nodes in the modular decomposition tree,
we use Lemma~\ref{lem:computeDistDisjointVertices} to compute for each pair of vertices in $G$ that are in different modules the length of a shortest path in constant time. Since we do this for each pair of vertices in $G$ exactly once, this sums up to a total running time of $\Oh(n^2)$.
The computation of the values $k_i$ can be done in time $\Oh(n_i^2)$, which is dominated by the time of solving vertex-weighted APSP.
In total, we obtain a combinatorial algorithm of time $\Oh(\mw^2n + n^2)$ and a algorithm of time $\Oh(\mw^{1.842}n + n^2)$ using fast matrix multiplication. 
\end{proof}

\section{Conclusion}\label{sec:Conclusion}

We started the study of \prob{Vertex-Weighted All-Pairs Shortest Paths} in the FPT in P framework and obtained efficient parameterized algorithms with respect to clique-width and modular-width. The algorithm parameterized by modular-width is adaptive, i.e., even if the parameter reaches its upper bound of $n$, the algorithm is not worse than the best unparameterized algorithm, and even for $k \in \Oh(n^{1-\varepsilon})$ for any $\varepsilon > 0$, it outperforms the best unparameterized algorithm. The algorithm parameterized by the stronger parameter clique-width is truly subcubic if $\cw \in \Oh(n^{0.5-\varepsilon})$ for any $\varepsilon >0$. It also permits us to solve \prob{Diameter} in the same time $\Oh(\cw^2n^2)$, complementing the lower bound ruling out $\Oh(2^{o(\cw)} \cdot n^{2 - \varepsilon})$ for any $\varepsilon > 0$, due to Coudert et al.~\cite{CoudertDP18}. The algorithms only apply to the vertex-weighted case. Note also that the algorithm relative to clique-width assume to be given a suitable expression or decomposition, whereas the modular decomposition of a graph, and hence its modular-width, can be computed in linear time~\cite{TedderCHP08}.

As mentioned in \cite{KratschN18}, considering edge-weighted graphs with (low) clique-width resp.\ low modular-width is hopeless, as one could modify an arbitrary input graph by adding all the missing edges with sufficiently large weights. Clearly, the shortest path lengths do not change, but the resulting graph is a clique and has constant clique-width and modular-width.

Apart from considering other parameters,
one interesting open question is whether there is an \emph{adaptive} algorithm for \APSP parameterized by clique-width, e.g., can the running time be reduced to $\Oh(\cw n^2)$? This seems quite challenging, since even computing some variant of \APSP for each node in the expression tree (on a graph with $\cw$ many nodes) results in a non-adaptive running time.

\bibliography{main}

\end{document}